%% file: main.tex
\documentclass[sigconf]{acmart}
\renewcommand\footnotetextcopyrightpermission[1]{} % removes footnote with conference information in first column
\usepackage{booktabs} % For formal tables
\usepackage{enumitem, amsmath, epsfig, graphicx, balance, subfigure, multirow, amssymb, hyperref, mathrsfs, color}
\usepackage[linesnumbered,ruled,vlined]{algorithm2e}
\newcommand{\xiuwen}[1]{{\color{black} #1}}

\newtheorem{theorem}{Theorem}
\newtheorem{example}{Example}

\newtheorem{lemma}{Lemma}
\newtheorem{definition}{\noindent Definition}

\def\BibTeX{{\rm B\kern-.05em{\sc i\kern-.025em b}\kern-.08emT\kern-.1667em\lower.7ex\hbox{E}\kern-.125emX}}
    
\copyrightyear{2019}
\acmYear{2019}
\setcopyright{acmlicensed}
\acmConference[SIGSPATIAL '19]{ACM International Conference on Advances in Geographic Information Systems}{November 5--8, 2019}{Chicago, IL}
\acmBooktitle{SIGSPATIAL '19: ACM Symposium on Neural Gaze Detection, November 5--8, 2019, Chicago, IL}
\acmPrice{15.00}
\acmDOI{}
\acmISBN{}

\begin{document}

%
% The "title" command has an optional parameter, allowing the author to define a "short title" to be used in page headers.
\title{Scalable Community Detection over Geo-Social Network}

%
% The "author" command and its associated commands are used to define the authors and their affiliations.
% Of note is the shared affiliation of the first two authors, and the "authornote" and "authornotemark" commands
% used to denote shared contribution to the research.

%\author{draft version}

\author{Xiuwen Zheng}
\affiliation{
  \institution{San Diego Supercomputer Center\\ University of California San Diego}
  %\streetaddress{P.O. Box 1212}
  \city{La Jolla}
  \state{CA}
  \postcode{92093}
}
\email{xiz675@eng.ucsd.edu}

\author{Qiyu Liu}
\affiliation{%
  \institution{The Hong Kong University of Science and Technology}
  %\streetaddress{1 Th{\o}rv{\"a}ld Circle}
  \city{Hong Kong SAR}
  \country{China}}
\email{qliuau@cse.ust.hk}

\author{Amarnath Gupta}
\affiliation{
  \institution{San Diego Supercomputer Center\\ University of California San Diego}
  %\streetaddress{P.O. Box 1212}
  \city{La Jolla}
  \state{CA}
  \postcode{92093}
}
\email{a1gupta@ucsd.edu}

% \affiliation{%
 
% \affiliation{%

% \affiliation{%



%
% By default, the full list of authors will be used in the page headers. Often, this list is too long, and will overlap
% other information printed in the page headers. This command allows the author to define a more concise list
% of authors' names for this purpose.

%
% The abstract is a short summary of the work to be presented in the article.
\begin{abstract}
We consider a community finding problem called \emph{\underline{C}o-located \underline{C}ommunity \underline{D}etection} (CCD) over geo-social networks, which retrieves communities that satisfy both high structural tightness and spatial closeness constraints. To provide a solution that  benefits from existing studies on community detection, we decouple the spatial constraint from graph structural constraint and propose a uniform CCD framework which gives users the freedom to choose customized measurements for social cohesiveness (e.g., $k$-core or $k$-truss). For the spatial closeness constraint, we apply the bounded radius spatial  constraint  and  develop an  exact algorithm together with effective pruning rules. To further improve the efficiency and make our framework scale to a very large scale of data, we propose a near-linear time approximation algorithm with a constant approximation ratio ($\sqrt{2}$). We conduct extensive experiments on both synthetic and real-world datasets to demonstrate the efficiency and effectiveness of our algorithms.
\end{abstract}

%   <concept_id>10010520.10010553.10010562</c% oncept_id>
%   <concept_desc>Computer systems % organization~Embedded % systems</concept_desc>
%   <concept_significance>500</concept_signif% icance>
%   <concept_id>10010520.10010575.10010755</c% oncept_id>
%   <concept_desc>Computer systems % organization~Redundancy</concept_desc>
%   <concept_significance>300</concept_signif% icance>
%   <concept_id>10010520.10010553.10010554</c% oncept_id>
%   <concept_desc>Computer systems % organization~Robotics</concept_desc>
%   <concept_significance>100</concept_signif% icance>
%   <concept_id>10003033.10003083.10003095</c% oncept_id>
%   <concept_desc>Networks~Network % reliability</concept_desc>
%   <concept_significance>100</concept_signif% icance>
 %</concept>
%</ccs2012>
%\end{CCSXML}

%\ccsdesc[500]{Computer systems organization~Embedded systems}

%
% Keywords. The author(s) should pick words that accurately describe the work being
% presented. Separate the keywords with commas.

% \keywords{Big Spatial Data, Computational Geometry, Location-Based Services, Spatial Data Mining and Knowledge Discovery}


%
% This command processes the author and affiliation and title information and builds
% the first part of the formatted document.
\maketitle

%%%%%%%% sections %%%%%%%%%
\input{exp_results/secs/introduction.tex}
\input{exp_results/secs/pre.tex}
\input{exp_results/secs/exact.tex}

\input{exp_results/secs/pruning.tex}

\input{exp_results/secs/approx.tex}

\input{exp_results/secs/experiment.tex}

\input{exp_results/secs/conclusion.tex}
\vspace{-1em}
%
% The acknowledgments section is defined using the "acks" environment (and NOT an unnumbered section). This ensures
% the proper identification of the section in the article metadata, and the consistent spelling of the heading.
% \begin{acks}
% \end{acks}

%
% The next two lines define the bibliography style to be used, and the bibliography file.
\bibliographystyle{ACM-Reference-Format}
\bibliography{my}

% 
% If your work has an appendix, this is the place to put it.
% \appendix

\end{document}

%% file: exp_results/secs/introduction.tex
\section{Introduction}\label{sec:introduction}
Finding densely connected structures in social networks, a.k.a., communities, has been extensively studied in past decades. Most of the prior research focus on  finding communities  in social networks \cite{fortunato2010community,  lancichinetti2009community, ahn2010link, huang2014querying}. However, some researchers  \cite{fang2017effective}, have argued that  for location-aware applications like location-based event recommendation and market advertisement,
 each community of people should be not only socially connected but also  be in close locational proximity to each other. Detecting such communities is called the \emph{\underline{C}o-located \underline{C}ommunity \underline{D}etection} (CCD)  problem.
%A typical community search algorithm with a user from New York as a query node may return a strongly connected community where all other members are  from cities other than NY. This community should  not  be a promising target  for a NY supermarket who wants to  promote their sales. %\todo{think of a better example.}
One reason for this increased emphasis on CCD problems is data availability -- the growing usage of mobile based services offered by social media applications that allow users to publish their real-time locations. In our own prior work \cite{desai2016pircnet,weibel2017hiv}, we investigated the formation of HIV related communities and determined that geographic proximity is a stronger predictor of community formation among users who tweet about HIV-related health issues compared to pure network proximity on Twitter.

 Some researchers have considered spatial location attributes to discover co-located communities  \cite{chen2015finding, chen2018maximum, fang2017effective, wang2018efficient}. We concluded these previous work and our work as Table~\ref{tab:prior_work} shows. These works can be classified into three categories based on their goals: 1)\underline{C}ommunity \underline{D}etection (CD): to find all co-located communities; 2) \underline{C}ommunity \underline{S}earch (CS) : finding personalized communities for query vertices; 3) \underline{M}aximum \underline{C}ommunity \underline{M}embers (MCM): find the maximum co-located community with the largest number of members, which is neither CS or CD. In this paper, we solve a community detection problem to find out all co-located communities. 
\vspace{-1.2em}
\begin{table}[htbp]
    \caption{Summary of existing studies.}
    %\vspace{-1em}
    \centering
    \footnotesize
    \begin{tabular}{|c|c|c|c|}
    \hline
     Spatial Constraint Def. & Algo. & CD or CS & Distance Bound\\ \hline\hline
     Modify Community Criteria &  Modified CNM \cite{chen2015finding}& CD &No guarantee \\\hline
     Minimum Covering Circle & AppAcc \cite{fang2017effective} & CS& No guarantee\\\hline 
    \multirow{3}{*}{Bounded Circle}&RotC$^{+}$ \cite{wang2018efficient} & CS &1\\
    & exact & CD & 1 \\
    & approx & CD & $\sqrt{2}$ \\ \hline 
    \multirow{3}{*}{All-pair Bound}& AdvMax \cite{zhang2017engagement}&  CD& 1\\
    &EffiExact\cite{chen2018maximum}&MCM& 1\\
    &Apx2\cite{chen2018maximum}& MCM& $\sqrt{2}+\epsilon$\\\hline
    \end{tabular}
    \label{tab:prior_work}
    \vspace{-1em}
\end{table}

%\underline{\textbf{Prior Work and Hardness.}}
% Discovering communities in social networks (the CD problem) has been well researched since its great importance. There are different measures to qualify the structural cohesiveness of a community. \textit{Modularity} is a popular global metric of community detection \cite{newman2004fast, white2005spectral, brandes2007modularity}, while \textit{Minimum degree} \cite{cui2014local, sozio2010community}, \textit{$k$-clique} \cite{cui2013online}, \textit{$k$-truss} \cite{huang2015approximate} are other common measures. Accordingly, for each measure  several techniques exist for the retrieval of highly-connected communities, such as core decomposition, truss decomposition and .
To give the definition of co-located community, we consider social and spatial cohesiveness constraints separately. There is already a significant body of research exists on community detection  and community search  on pure social network \cite{newman2004fast, white2005spectral, brandes2007modularity, cui2014local, sozio2010community, cui2013online, huang2015approximate}, 
we focus primarily on the  spatial constraint aspect of the CCD problem. The existing work can  be categorized into  four types based on their approaches to defining spatial constraint. The first technique is to define new criteria of community by integrating both social  and spatial information.  For example, \cite{chen2015finding} modifies the modularity function \cite{newman2004fast} by introducing a distance decay function and then provides a community detection algorithm based on fast modularity maximisation algorithm. There are two main limitations of this technique. The first and the most serious one is that it can not provide a  geographic distance bound guarantee for members in a community. Secondly, it couples social and spatial information which is less flexible if users prefer other community detection techniques, e.g. $k$-core or  $k$-truss decomposition. 
% With much prior work in  defining   and  discovering highly connected structure in social network, the hardness of finding co-located communities lies in the spatial cohesiveness constraint.

Different from the first technique, other existing works define  spatial constraint without any regard to  social cohesiveness. Fang et al \cite{fang2017effective} apply the spatial minimum covering circle \cite{elzinga1972minimum, elzinga1972geometrical} to ensure that each community  discovered  maintains high spatial compactness. It provides a community search algorithm that returns a $k$-core  so that  a  spatial circle with the smallest radius can cover all community members. However,  it  can not guarantee a consistent distance bound for different query vertices and the case study in experiments section well presents an example to demonstrate it. Some  research \cite{zhang2017engagement, chen2018maximum}   apply the all-pair distance constraint which  requires that the distance between any two users in a community is within a user-specified threshold. They can  guarantee a  bounded distance constraint for all members, however, \cite{zhang2019efficient} proves that the problem to enumerating all maximal clusters satisfying this spatial constraint is NP-hard. \cite{zhang2017engagement, chen2018maximum} provide a clique-based algorithm to enumerate maximal clusters and then find $k$-cores in each cluster. Since they solve the MCM problem where only \emph{the maximum} community is returned, they develop pruning rules when enumerating clusters. However, it is impractical to introduce it to solve our CD problem where \emph{all} communities should be detected. Even though \cite{chen2018maximum} provides an approximation algorithm based on grids, the approximation ratio is not a constant.

% \begin{figure}[th]
%     \centering
%     \includegraphics[width=0.2\textwidth]{figs/alice.pdf}
%     \caption{A bad case of the smallest covering circle.}
%     \label{fig:cc}
% \end{figure}

%\cite{chen2018maximum}\cite{zhang2017engagement} adopts all-pair distance constraint.  They can guarantee that for an returned community, the spatial distance between any two members are within a user specified constant. In \cite{zhang2017engagement}, each of the returned communities  should be a $k$-core and satisfy pairwise similarity constraints where the similarity can be distance. It proves the NP hardness of this problem by reducing a NP-complete problem, finding $k$-clique, to it, and then provides a clique-based algorithm.
%\cite{chen2018maximum} applies the same methodology to solve the problem where similarity is defined as spatial distance in $\mathbb{R}^2$ space. They provide the clique-based algorithm as baseline and  also provide an approximation algorithm based on grid. However, for their approximation algorithm, the approximation ratio is not a constant and is related to the grid size. 
% $\sqrt{2}+\epsilon$ and $\epsilon = \frac{2\sqrt{2}w}{d}$ where $w$ is the grid size and $d$ is the pairwise distance threshold. When $w \ll d$, the approximation ratio is small, which is desirable, however, it would require more memory to store grid information; when grid is set to be a large value, then approximation ratio would be large.  

\cite{wang2018efficient} defines spatial constraint in a similar way to our work. It applies radius-bounded circle approach which requires that any returned  community can be covered geographically by a circle with user-specified radius, and it uses $k$-core to ensure social cohesiveness. It solves a community search problem which returns co-located communities for a specified query user.  Different from the problem in \cite{chen2018maximum}\cite{zhang2017engagement} which is NP-hard, the problem in   \cite{wang2018efficient} can be solved in polynomial time. 
However, the approach in \cite{wang2018efficient} cannot be directly modified to solve our CD problem because of time complexity issue. Given an query user, it enumerates  all candidate circles with user-specified radius passing  any two nearby users and finds the $k$-core containing query user in each circle. For any query user, there can be $O(n^2)$ candidate circles where $n$ is number of users, and thus calculating all nodes enclosed by each circle and then finding cores from them can be very time-consuming. \cite{wang2018efficient} solves a community search problem, thus $n$ is practically a small value and time complexity is not a severe issue. However, in our problem, more efficient methods should be proposed.

\noindent \textbf{Contributions.} Our major contributions are listed as follows.
\begin{itemize}
\item We design a uniform framework which decouples social and spatial constraints so that users have high freedom to define social cohesiveness  (e.g., $k$-truss or $k$-core) and to choose  existing community detection algorithms.
    \item We design an  exact  spatial constraint checker which return all maximal communities satisfying spatial constraint efficiently.
    \item To further reduce the time complexity, we design a near-linear approximation algorithm with a constant performance guarantee ($\sqrt{2}$-bounded) for the spatial constraint checker.
    \item Our framework can be modified slightly to solve community search problem for any given query user. 
\end{itemize}

%% file: exp_results/secs/pre.tex
\section{Preliminaries}\label{sec:pre}
In this section, we formally define our data model and problem,  and present the framework for co-located community detection. 

\subsection{Problem Definition}\label{subsec:problem_def}
\begin{definition}[Geo-Social Network (GeoSN)]
A geo-social network (GeoSN) is a directed graph $G=(V, E)$ where each $v\in V$ denotes a user  associated with a spatial location $(x_v, y_v)\in \mathbb{R}^2$, and $E$ maintains the relationship (e.g., friendship) among users. 
\end{definition}

Given a geo-social network, the objective of this paper is to find all communities that simultaneously satisfy the spatial cohesiveness  constraint and the social connectivity constraint. 
We first introduce the definition of a maximal co-located community. 
\begin{definition}[Maximal Co-located Community (MCC)]
Given a GeoSN $G$, a maximal co-located community is a set of users which form a  subgraph $J\subseteq G$ satisfying three constraints,
\begin{itemize}
\item \textbf{Social Connectivity:} $J$ should satisfy a user-specified social constraint  over a graph property like  $k$-truss, $k$-core, etc.
\item \textbf{Spatial Cohesiveness:}  Given a distance threshold $d$, all the vertices  of $J$ can be geographically enclosed by a circle with diameter $d$.
\item \textbf{Maximality:} There does not exist a subgraph $J' \supsetneq J$ which satisfies social connectivity and spatial cohesiveness constraints.  
\end{itemize}
\end{definition}

The following formally defines the  \textit{$d$-MCCs Detection} problem and presents an example,
%\vspace{-5pt}
\begin{definition}[$d$-MCCs Detection] \label{def:mcc}
Given a geo-social network $G$,  a distance threshold $d$ and social constraint, the problem is to find all maximal co-located communities.  
\end{definition}

\begin{example}
Fig.~\ref{fig:gsn} (a) presents a Geo-social network where users are denoted as circles and relationships are denoted as lines. Each user is associated with a location in $\mathbb{R}^2$ space. Suppose that high social cohesiveness is defined as a minimum degree of at least 2, then there are two communities found in the GeoSN denoted as blue circles  and orange circles respectively. Suppose that the distance threshold is set as 4 grids, then users can be divided into four overlapped groups based on their locations denoted as four shadow circles in Fig.~\ref{fig:gsn} (b).  Combining spatial and social information, there are two MCCs detected: $\{a, b, c, d\}$ and $\{i,j,k,l\}$. 
\end{example}

\begin{figure}[htbp]
\centering
\includegraphics[width=0.4\textwidth]{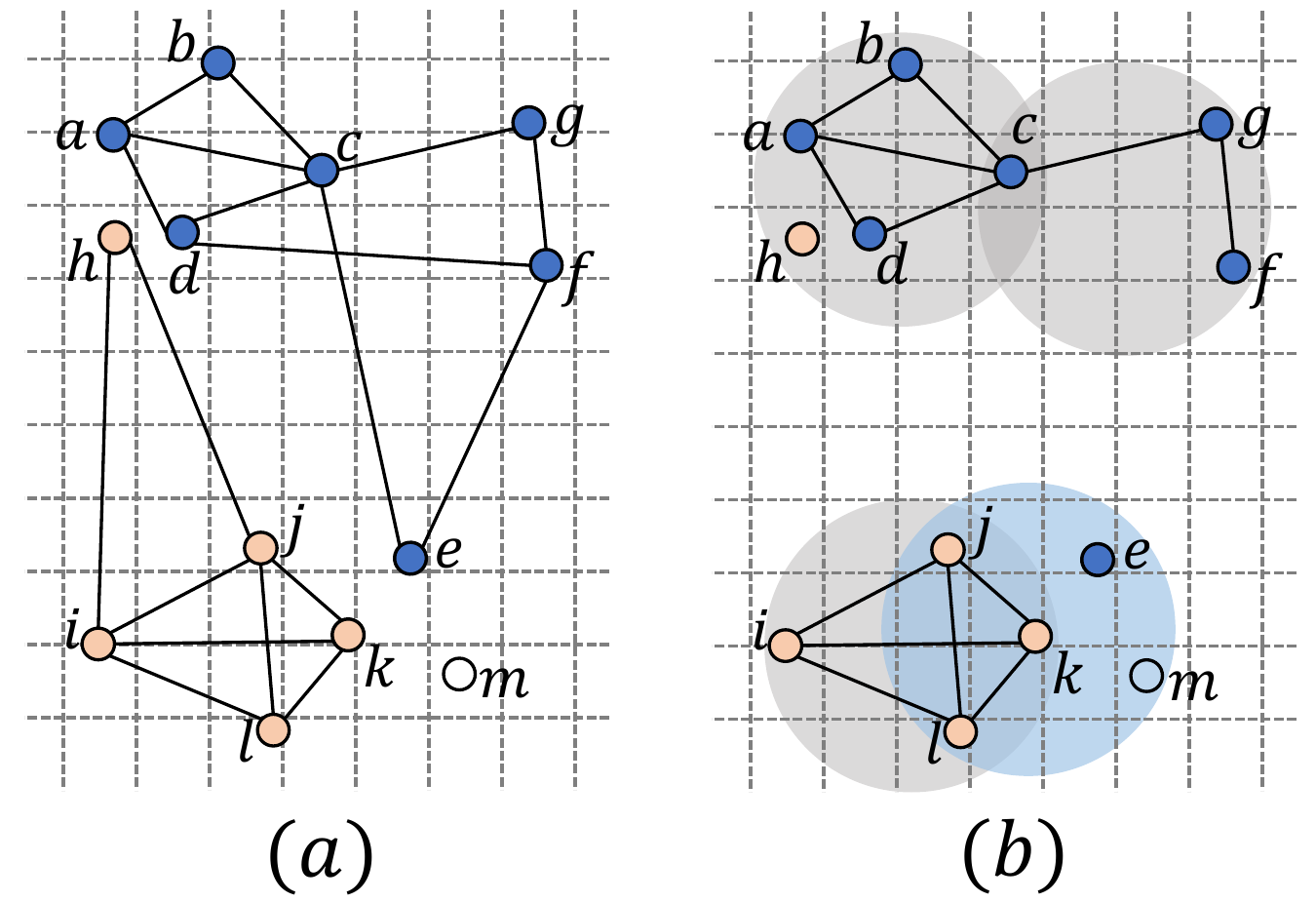}
%\vspace{-0.5em}
\caption{Illustration of geo-social network and MCCs.}\label{fig:gsn}
\end{figure}
\subsection{Framework}\label{subset:pre_for_exact_algo}
%Naturally, we could consider spatial and social  constraints  separately. As stated in Section \ref{sec:related}, there are many existing metrics, e.g., $k$-truss, to measure social connectivity and many existing methods to detect communities satisfying these measures. Thus, we will firstly focus on spatial constraints.

To de-couple the  spatial constraint from MCCs detection, we provide the definition of  Global Spatial Cluster  merely based on spatial constraints.
\begin{definition}[Global Spatial Cluster]
A Global Spatial Cluster (GSC) $C\in V$ is a subset of users satisfying two constraints,
\begin{itemize}
\item{\textbf{Co-located:} The cluster members of $C$ satisfy the spatial constraint (i.e., being enclosed by a circle with diameter $d$).}
\item{\textbf{Maximality:} If $C$ is an GSC, there does not exist a set of users $C'\supsetneq C$ satisfying the co-located  constraint. }
\end{itemize}
\end{definition}
Based on this definition,  MCCs can be detected in three steps (Algorithm~\ref{algo:framework}):
(1)  find all Global Spatial Clusters  (line 1), (2) for each GSC, get the GeoSN induced by this set of users and find all local MCCs in this subgraph based on social constraint (lines 2 - 4), (3) find all MCCs by removing local MCCs which are subsets of any other MCC (FindGlobalMCC function).
 Note that in the first step (line 2),  i.e, finding GSCs, the parameters of social constraint are also passed to the spatial algorithm so that some simple pruning techniques can be implemented. 
%We call social community detection  algorithm on the whole GeoSN in line 1 and  also on each global maximal set in line 5 for different purpose. Line 1 is to prune space for   detecting global maximal set so that we only need to detect GSC in each community which contains a much smaller set of  members than the whole Geo-social network, and line 5 is to return communities satisfying the social constraint in each GSC.

%\xiuwen{\subsection{All $d$-MCCs Search}The previous subsections accomplish the first step of the baseline, i.e, find GSCs. Here we present the whole algorithm to detect all $d$-MCCs. Note that for the ease of demonstration, we use $k$-core as the social constraint measure. In the experiments, we implement both $k$-core and $k$-truss.}
\begin{algorithm}
\footnotesize
\caption{\emph{Framework\_MCCs\_Detection}}\label{algo:framework}
\KwIn{A  geo-social network $G = (V, E)$, distance threshold $d$, social constraint parameter $k$ }
\KwOut{A set of all MCCs: $MCC$}
\SetKwFunction{FindGlobalMCC}{FindGlobalMCC}
\SetKwProg{Fn}{Function}{:}{}
%\tcc{apply existing community detection algorithm on whole GeoSN}
%$C \gets$  \texttt{CommunityDetection}$(V, G, k)$\;

%\For{$c \in C$}{
\tcc{apply spatial algorithm to get GSCs}
$GlobalMS\gets$ \texttt{SpatialAlgorithm}$(V, d, k)$\;
$local\_mcc \gets \{\}$\;
\tcc{apply CD algorithm on the subgraph induced by each GSC}
\For{$GSC$ in $GlobalMS$}{
$local\_mcc.addAll($ \texttt{CommunityDetection}$(GSC, G, k))$\; 
}
\tcc{find global MCCs}
$MCC \gets$ \FindGlobalMCC{$local\_mcc$}\;
%}
\textbf{return} $MCC$\;

% \Fn{\FindGlobalMCC{$Local\_MCCs$}}{
% $Local\_MCCs.sort(key = lmcc.length(), reverse = true)$\;
% $MCCs = \{ \}$\;
% \For{$lmcc$ in $Local\_MCCs$}{
% $MCCs.add(lmcc)$ if no set in $MCCs$ contains $lmcc$}
% \textbf{return} $MCCs$\;
% }
\end{algorithm}
\setlength{\textfloatsep}{0pt}

\begin{example}
We still take the GeoSN in Fig.~\ref{fig:gsn} (a) as an example and keep the same constraint definitions to illustrate the procedures. The first step returns four GSCs detected as the shadow circles in (b) show. In the  social subgraph induced by vertices in each GSC, detect the local MCCs based on social constraint, then we get three sets: $\{a, b, c, d\}$, $\{i, j, k, l\}$, and $\{j, k, l\}$. By calling the function FindGlobalMCC, $\{i, k, l\}$ covered in the blue shadow circle is removed from MCCs. Thus, we detect two MCCs: $\{a, b, c, d\}$ and $\{i, j, k, l\}$.
\end{example}
By adding two more procedures to this framework, we can easily solve the corresponding  community search problem for given query user. Lines 1 and 2 find all candidate nodes that are within euclidean distance $d$ from query node $q$ and extract  the small subgraph formed by this set of nodes. After applying Algorithm~\ref{algo:framework} to get all MCCs in the subgraph, line 4  filters out the MCCs that does not contain $q$. 
\vspace{-1em}
\begin{algorithm}
\footnotesize
\caption{\emph{Framework\_MCCs\_Search}}\label{algo:framework2}
\KwIn{A  geo-social network $G = (V, E)$, query user $q$, distance threshold $d$, social constraint parameter $k$ }
\KwOut{A set of $q$-MCCs: $qMCCs$}
\SetKwFunction{FrameworkMCCsDetection}{FrameworkMCCsDetection}
\SetKwFunction{FindGlobalMCC}{FindGlobalMCC}
\SetKwProg{Fn}{Function}{:}{}
%\tcc{apply existing community detection algorithm on whole GeoSN}
%$C \gets$  \texttt{CommunityDetection}$(V, G, k)$\;

%\For{$c \in C$}{
\tcc{get nearby users to form a small network}
 $V' = \{u| ed(u, q) \leq d\}$, $E' = \{(u_1, u_2)|u_1, u_2 \in V',  (u_1, u_2)\in E\}$\;
 $G' \gets (V', E')$\;
 \tcc{apply MCCs detection framework to get candidate MCCs}
$MCCs \gets$ \FrameworkMCCsDetection$(G', d, k)$\;
$qMCCs \gets$ \FindGlobalMCC{$\{mcc \mbox { for }mcc \in MCCs \mbox{  if }  q \in mcc\}$}\;
\textbf{return} $qMCCs$\;

% \Fn{\FindGlobalMCC{$Local\_MCCs$}}{
% $Local\_MCCs.sort(key = lmcc.length(), reverse = true)$\;
% $MCCs = \{ \}$\;
% \For{$lmcc$ in $Local\_MCCs$}{
% $MCCs.add(lmcc)$ if no set in $MCCs$ contains $lmcc$}
% \textbf{return} $MCCs$\;
% }
\end{algorithm}
\setlength{\textfloatsep}{0pt}
\vspace{-1.5em}

The following three sections focus on developing spatial algorithms to detect GSCs, and in the experiment section, we apply Algorithm~\ref{algo:framework} with $k$-core or $k$-truss as social constraint to detect all MCCs on five datasets.

%% file: exp_results/secs/exact.tex
\section{Exact Spatial Algorithm}
% \subsection{Baseline Algorithm}
%\qliu{As we have shown in last section, listing all possible GSCs does not in NP-hard. 

In this section, we will give an exact algorithm  for detecting all Global Spatial Clusters in $\mathbb{R}^2$ space.
The basic idea is to transform the input spatial space from Cartesian coordinate system to polar coordinate system, and based on which an angular sweep procedure is repeatedly invoked for each node to ensure that no GSC is missed.

\subsection{Local Spatial Cluster}
Global Spatial Cluster is defined based on covering circle with fixed radius, and the following will give the definition of  a more restricted covering circle, called $v$-bound circle, and based on which Local Spatial Cluster (LSC) is defined.

\begin{definition}[$v$-bounded circle]\label{def:bounded_circle}
Given a point $v\in\mathbb{R}^2$, if a circle with diameter $d$ (user-specified distance threshold) passes $v$, then it is called a $v$-bounded circle denoted as $C_v$. 
\end{definition}
For a given point $v$, set it as  reference point and $x$-axis as  reference direction to build a polar coordinate system. If the center of a $v$-bounded circle $C_v$ has coordinate $(r, \theta)$ in this polar coordinate system, denote this circle as $C_v(r,\theta)$\footnote{We alternatively use $C_v(r, \theta)$ and $C_v$ if the context is clear.} where $\theta \in (-180^{\circ}, 180^{\circ}]$. Now we give the definition of $v$-Local Spatial Cluster as follows.

\begin{definition}[$v$-Local Spatial Cluster (LSC)]\label{def:LSC}
Let $r = d/2$, a $v$-Local Spatial Cluster $L_v$ is a set of points enclosed by circle $C_v(r, \theta)$ such that there does not exist a  circle  $C_v(r, \theta')$ enclosing  a proper superset of $L_v$. Denote the set of all $v$-LSCs for a fixed $v$ as $\mathcal{L}_v$.
\end{definition} 
We then have the following lemma showing the relationship between global and local  spatial clusters, which is the backbone of the exact spatial algorithm.

% The algorithm for finding all GSCs runs in three steps and the first two steps find out all Local Spatial Clusters and the last step generates all GSCs from LSCs. Here is a lemma that illustrates the correctness of generating all GSCs from LSCs. 

\begin{lemma}\label{lemma:correctness}
Given a set of points $V$ in $\mathbb{R}^2$ and a distance threshold, denote the set of all GSCs as $\mathcal{G}$, then  $\mathcal{G}\subseteq \cup_{v\in V}\mathcal{L}_v$.
\end{lemma}

\begin{proof}
 As shown in Fig.~\ref{fig:lemma2}, let all the small  circles consist  a GSC $G$,  by definition, there is a circle with radius $r = d/2$ covering them (shown as the large  black circle).   Let $ed(u, v)$ denote the spatial distance (Euclidean distance) between two users.  W.o.l.g., assume that $a$ is the farthest point  in the circle from  $O$ and then  the dashed  circle centered at $O$ with radius $r' = ed(a, O) < r$ still encloses all points in $G$. Find a point $Q$ on the line $\overrightarrow{aO}$ such that $ed(a, Q) = r$ and  get the grey circle centered at $Q$ with radius $r$ shown as the grey circle. For any point $b\in G$, based on triangle inequality, $ed(b, Q) \leq  ed(b, O) + r-r'  \leq r$. Thus all points in $G$ can be covered by the $a$-bounded grey circle, i.e.,  $G$ is an $a$-LSC. 
\end{proof}

 \begin{figure}[htbp]
 \centering
 \includegraphics[width=0.2\textwidth]{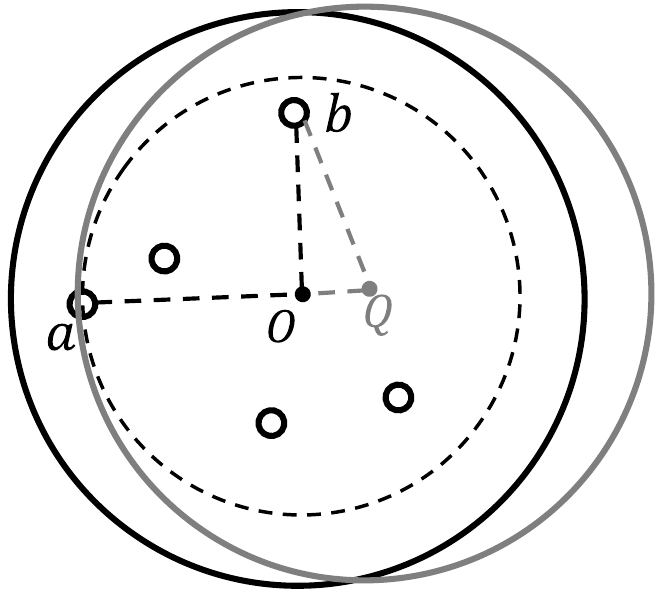}
 %\vspace{-0.6em}
 \caption{Illustration of Lemma~\ref{lemma:correctness}.}\label{fig:lemma2}
 \vspace{-1.7em}
\end{figure}
\subsection{Searching All LSCs}\label{subsec:find_LSC}

Lemma~\ref{lemma:correctness} shows that the problem of finding all GSCs can be solved by calculating $v$-LSCs $\mathcal{L}_v$ for every $v\in V$. To efficiently calculate $\mathcal{L}_v$ for a given reference point $v$, in this subsection, we introduce  the \emph{Angular Sweep}-based technique.

% Given coordinates of each user point in $\mathbb{R}^2$ space, we arbitrarily select one point, denoted as $O$. Before finding $O$-LSCs, we first determine the range of $\theta$ so that $C_O(r, \theta)$ covers any specific point. $C_O(r, \theta)$ 

Suppose that circle $C_v(r, \theta)$ rotates counterclockwise, i.e., $\theta$ increases from $-180^{\circ}$ to $180^{\circ}$,  for each point within distance $d = 2r$ from $v$, we consider two special events: it first enters $C_v$ and it quits $C_v$, and we call the angles $\theta$ at these two special events as start angle $\theta.start$ and end angle $\theta.end$ respectively. When $\theta \in [\theta.start, \theta.end]$, the circle $C_v(r, \theta)$ always encloses  this point. Figure~\ref{fig:rotate-illustrate} illustrates such rotation process.
\begin{figure}[t]
\hspace{-0.2in}
    \begin{minipage}[t]{0.4\linewidth}
        \centering
        \includegraphics[width=\textwidth]{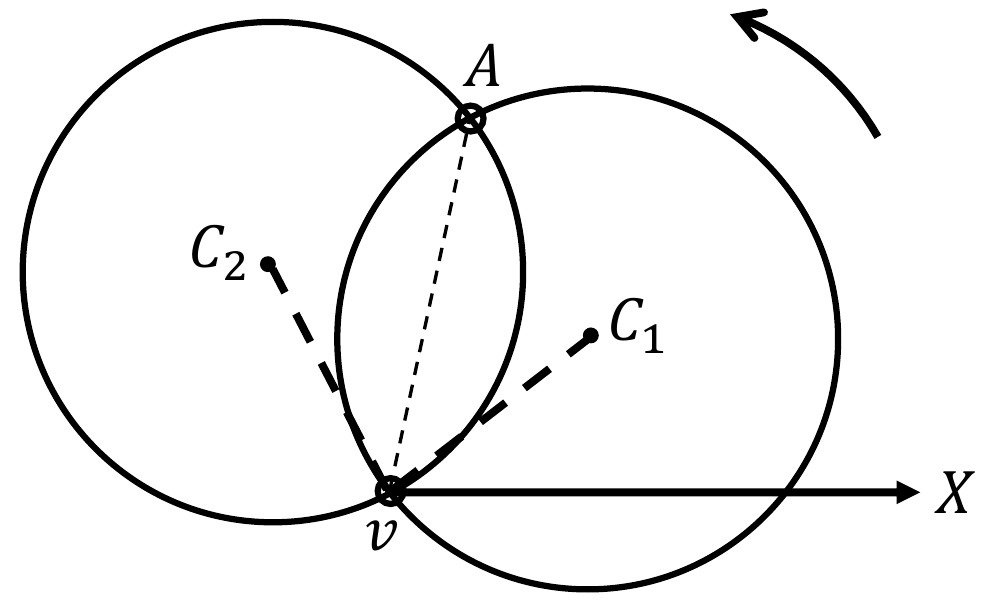}
               % \vspace{-1em}
       %\setlength{\belowcaptionskip}{-2pt}
        \caption{\small{two events for $A$.}}\label{fig:rotate-illustrate}
    \end{minipage}
        \hspace{0.2in}
    \begin{minipage}[t]{0.4\linewidth}
        \centering
    \includegraphics[width=\textwidth]{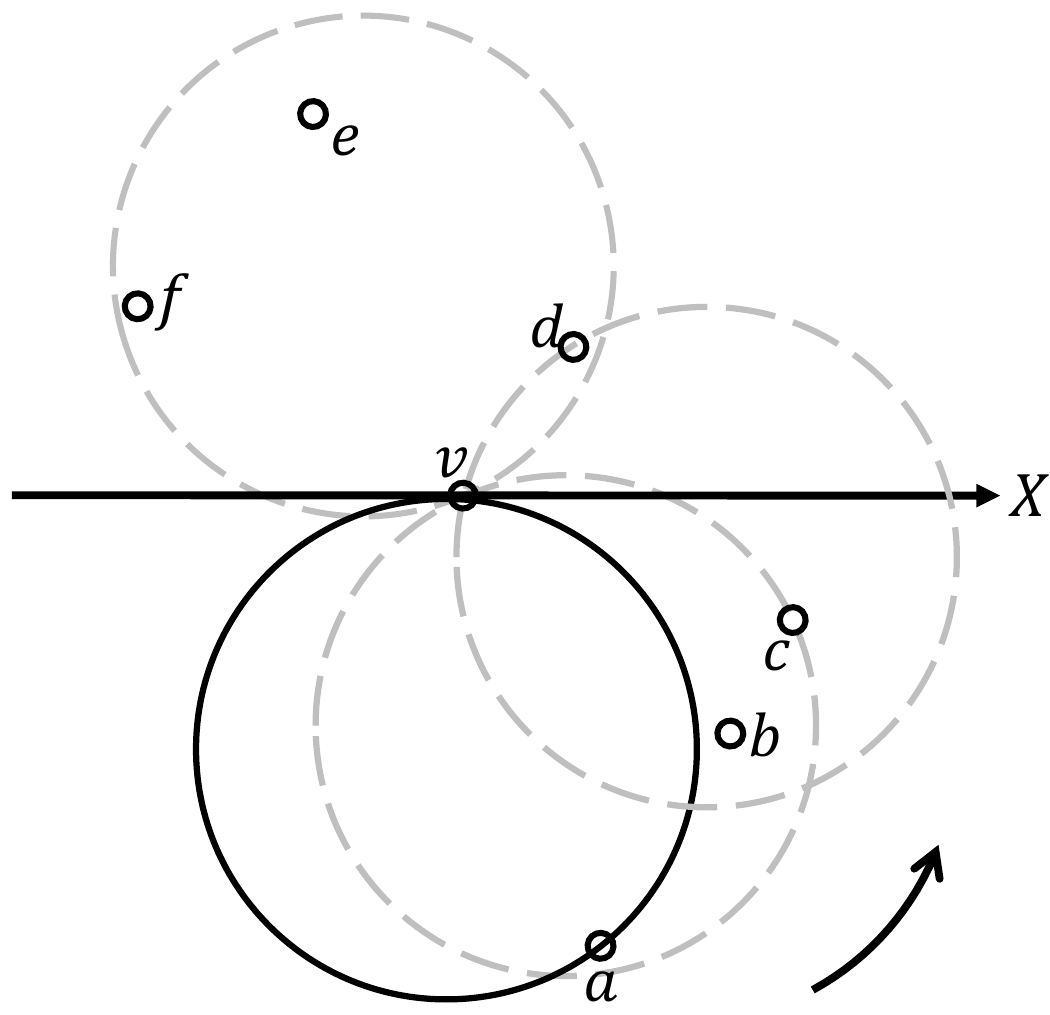}
    %\vspace{-1em}
    %\setlength{\belowcaptionskip}{-2pt}
    \caption{\small{Angular Sweep}}
    \label{fig:angular}
    \end{minipage}
%\vspace{-1em}
\end{figure}
%\vspace{-1em}
% \begin{figure}[htbp]
% \centering
% \includegraphics[width=0.25\textwidth]{figs/ill2.pdf}
% \caption{Illustration of Lemma~\ref{lemma:correctness}.}\label{fig:lemma2}
% \end{figure}
% \begin{figure}[htbp]
% \centering
% \includegraphics[width=0.3\textwidth]{figs/ill.pdf}
% \caption{Start and end events for point $A$ when rotate a circle with reference point $v$.}\label{fig:rotate-illustrate}
% \end{figure}

In Figure~\ref{fig:rotate-illustrate}, the  circles centered at $C_1$ and $C_2$ are  $\mathcal{C}_v$ at the two special events for point $A$. Denote polar coordinates of $C_1$ and $C_2$ as $(r, \theta.start)$ and $(r, \theta.end)$ respectively, and polar coordinate of $A$ as $(d_A,\alpha_A)$, then we calculate $\theta.start$ and  $\theta.end$ using equations:
\begin{align}
\vspace{-1em}
\label{eq:theta1}\theta.start &=  \alpha_A - \cos^{-1} \frac{d_A}{2r}\\
\label{eq:theta2}\theta.end &=  \alpha_A + \cos^{-1} \frac{d_A}{2r} 
\end{align}
% \textbf{Get All Local Spatial Clusters (LSC) With One Node As Reference Point:}
\noindent Given a reference point $v$ and a set of vertexes $V'\subset V$ where each vertex is within $d$ from $v$, Algorithm~\ref{algo:local_maximal_set} outputs all $v$-Local Spatial Clusters.
Lines 1 - 5 first calculate start and end angles for nodes in $V'$ via  Eq.~\eqref{eq:theta1} and Eq.~\eqref{eq:theta2} and sort nodes  based on start angles. Lines 6 - 19 present the angular sweep procedure (Fig.~\ref{fig:angular}).

Let the initial state of $C_v(r, \theta)$ (shown as the black circle in Fig.~\ref{fig:angular}) be at the place where it just passes the first node ($\theta = a.start$ ) and let the candidate set $CS = \{a\}$ which records the set of points currently enclosed by $C_v$.  Let $end$ keep track of the smallest end angle of points in $CS$. Keep rotating $C_v$ counter-clockwisely to the next points and adding new  points  to $CS$ until one  point in $CS$  will leave $C_v$.  More specifically,
denote the next point that $C_v$ is going to reach as $x$, when $x.start > end$ which means that at least a point is going to leave $C_v$,  add $CS$ to LSC set. Rotate $C_v$ to reach $x$, add $x$ to $CS$ and remove points whose end angles are less than $x.start$ to form a new candidate set. Keep the above procedure until reaching the last point.  For example, in Fig.~\ref{fig:angular}, add $b, c$ to $CS$ step by step and then when $C_v$ is going to  enclose $d$, since $d.start > end =  a.end$, the current set $CS = \{a, b, c\}$ should be a LSC. Then remove point $a$ from $CS$ because it has left $C_v$ and add $d$ to $CS$. Keep rotating and generating LSC until the circle encloses the last point $f$. There are three LSCs detected as the grey dashed circles enclose.
\begin{algorithm}
\footnotesize
\caption{\emph{Local\_Spatial\_Clusters}}\label{algo:local_maximal_set}
\KwIn{Reference node $v$, a vertex set $V'$}
\KwOut{A set of $v$-LSC: $LSC$}
\tcc{use $v$ as reference point and $x$ axis as direction to build polar coordinate system}
Interval list $P_v \gets [\ ]$\;
\For{$(x_u, y_u)$ in $V'$}{
		calculate  $\Theta_u.start$, $\Theta_u.end$\;
		$\Theta_u.node = u$, $P_v.add(\Theta_u)$\;
		}
		
$I \gets P_v.sort(key=\Theta.start)$\;
% Sort $P$ in the increasing order of $P_i.start$ as $I$\;
$LSC \gets \{ \}$, $CS = \{I_1\}$\;
$end \gets I_1.end$, $idx \gets 2$\;
\While{$idx \leq I.length$}{
	\If{$I_{idx}.start \leq end$}{
	$CS.add(I_{idx})$\;
	$end \gets \min(end, I_{idx}.end)$\;
	}
	\Else{
	\For{$\Theta$ in $CS$}{
		%$NodeSet.add(\Theta.node)$\;
		\If{$\Theta.end < I_{idx}.start$}{
		$CS.remove(\Theta.node)$\;
		}
	}
	$CS.add(I_{idx})$\;
	$LSC.add(CS)$\;
	$end \gets \min(\Theta_{i}.end)$ for $\Theta_{i}\in CS$\;
	}
	$idx \gets idx+1$\;
}
\lIf{$CS$ not empty}{$LSC.add(CS)$}
\textbf{return} $LSC$\;
\end{algorithm}
\setlength{\textfloatsep}{0pt}
% Scan the list $I$ from left to right. 

%As \qliu{shown in lines} 10-23, for each interval $I_{idx}$, if its start angle \qliu{$I_{idx}.start$} is smaller than $temp$, safely add it to $CS$ to grow the candidate set and update $temp$ if needed; otherwise, its corresponding nodes set $NodeSet$  should be an $v$-LSC and be added into $LocaLSC$, and $CS$ would be updated as a new candidate set for the next $v$-LSC: remove any interval $\Theta$ from $CS$ if $\Theta.end < I_{idx}.start$ and add $I_{idx}\textbf{}$ to $CS$. Then $CS$ would be a candidate set for the next $v$-LSC. The algorithm terminates when all intervals in $I$ have been processed, \qliu{and $SC$ will be added to $LocaLSC$ if it is not empty after the loop at line 18.}

 \textbf{Complexity Analysis.} %Suppose there are $n$ users in the geo-social network. In the worst case, all points are within $d$ from the reference point $v$, which results in  the input interval list having length $n$. 
Suppose that the input vertex set $V'$ has a size $m$, then
Line 5 takes time $O(m\log m)$ by using a conventional sorting algorithm. For the angular sweep  shown in lines 6-19, the update of candidate set $CS$ (lines 13 to  15), which dominates the loop body, is executed in $O(m)$ time. Thus, the total worst case time complexity of Algorithm~\ref{algo:local_maximal_set} should be $O(m\log m + m^2)=O(m^2)$. For any vertex $v$, the number of $v$-LSCs is $O(m)$.
%Notice that in practice, spatial threshold $d$ is a small value, and the number of points in each Local Spatial Cluster can be regarded as a constant, so $CS$ can be updated in constant time, thus Algorithm~\ref{algo:local_maximal_set} run in $O(n)$ time. 

\subsection{Searching GSC}\label{subsec:find_GSC}
An LSC may not be a GSC as it might be a subset of another LSC with a different reference node. Thus, by excluding  any LSC which is a subset of another LSC, we obtain all GSCs.

The whole algorithm to find GSCs  is presented in Algorithm \ref{algo:global_maximal_set}. Note that for a certain social constraint, e.g., $k$-core, $k$-truss, some simple pruning can be implemented to reduce search space. Algorithm \ref{algo:global_maximal_set} uses $k$-core as an illustration. In the experiments, we implement both $k$-core and $k$-truss.
 For each node $v$, to reduce search space, line 3 applies range query to find out vertexes within distance $d$ from the location of $v$ since any vertex outside this circle  can not be in a $v$-LSC. Since we need to find $k$-core at last, if the number  of vertexes lie in the circle is less than $k$, these LSCs can not contain any MCC and we skip them as line 4 shows.
Line 5 invokes  Algorithm~\ref{algo:local_maximal_set} \emph{LocalMaximalSet} to find out all $v$-LSCs. After detecting all LSCs, the function \emph{FindGSC} is invoked to add  LSCs which are not subset of any others to   the GSCs set $GSC$.
\begin{algorithm}
\footnotesize
\caption{\emph{Global\_Spatial\_Clusters}}\label{algo:global_maximal_set}
\KwIn{A set of nodes $V$ of a GeoSN,  distance threshold $d$, social constraint $k$}
\KwOut{A set of all Global Spatial Clusters (GSCs) $GSC$}
\SetKwFunction{FindGSC}{FindGSC}
\SetKwFunction{LocalSpatialClusters}{LocalSpatialClusters}
\SetKwProg{Fn}{Function}{:}{}
$LSC = \{ \}$\;
\For{$v$ in $V$} {
	\xiuwen{\tcc{do a range query to find all nodes within $d$ distance to reference node.}
	$CV \gets$ range\_query(v, d)\;
	\lIf{$CV.size() < k$}{continue}}
	\xiuwen{\tcc{detect all $v$-LSC}}
	$\mathcal{L}_v\gets$ \LocalSpatialClusters{v, $CV/\{v\}$}\;
    $LSC$.add($\mathcal{L}_v$)\;
	}
	$GSC \gets$ \FindGSC{$LSC$,  $k$}\;
\textbf{return} $GSC$\;

 \Fn{\FindGSC{$LSC$,  $k$}}{
 $LSC.sort(key = LSC.length(), reverse = true)$\;
 $GSC = \{ \}$\;
 \For{$lsc$ in $LSC$}{
 \lIf{$lsc.size() < k$}{continue}
 $GSC.add(lsc)$ if no set in $GSC$ contains $lsc$}
 \textbf{return} $GSC$\;
 }
\end{algorithm}
\setlength{\textfloatsep}{0pt}

\textbf{Complexity Analysis.} Assume that there are $n$ vertexes in GeoSN, i.e., $|V| = n$,
in the worst scenario, for each vertex $v\in V$, there are  $O(n)$ vertexes within distance $d$ to $v$,  and thus the worst time complexity for finding  $v$-LSCs  (line 5) would be $O(n^2)$ as analyzed in last subsection. Thus, finding all LSCs would cost $O(n^3)$. There are $O(n^2)$ LSCs  in total, thus  function  \emph{FindGSC} will do $O(n^4)$ set comparisons  where each single comparison takes time $O(n)$.  The total time complexity in the worst case should be $O(n^3 + n^5) = O(n^5)$.
However, in practice, the spatial threshold $d$ is a small value. Assume that the location density of points is $\rho$, and let $C_x = \rho \pi (x/2)^2$,  then it takes time $O(C_{2d}^2)$ to get $v$-LSCs,  the number of $v$-LSCs would be $O(C_{2d})$ and each set has points $O(C_d)$, so 
the time complexity would be $O(nC_{2d}^2 + n^2C_{2d}^2C_d) = O(n^2C_{2d}^2C_d)$.
%thus for each vertex $v$ the size of $CV$ can be regarded as a constant number and thus the  number of $v$-LSCs and the size of each LSC would also be constant. Finding $v$-LSCs takes constant time and  detecting all LSCs would take $O(n)$.}  Since there are only $O(n)$ LSCs,  only $O(n^2)$ times set comparisons  each of which  takes constant time needed to conducted.  The total time complexity would be $O(n + n^2) = O(n^2)$.

%% file: exp_results/secs/pruning.tex
\section{Pruning and Optimization}\label{sec:pruning}
 The high time complexity of Algorithm~\ref{algo:global_maximal_set} 
%, $O(n^5)$ for worst case and $O(n^2)$ for practical case,  
in last section
prevents it being scaled to large dataset. Thus, we propose several pruning strategies and optimization tricks for Algorithm~\ref{algo:global_maximal_set}, which is experimentally demonstrated to accelerate the algorithm a lot and reduce time by orders of magnitude in some datasets.

In Algorithm \ref{algo:global_maximal_set}, in the worst case an LSC needs to be compared with  other $O(n^2)$ LSCs to determine whether or not it is a GSC, which is extremely inefficient and is  the dominant part of the time complexity. In this section, we develop pruning rules to dramatically reduce the times of set comparisons.

\textbf{Pruning rule I: point-wise pruning.}
 Given an $a$-LSC and a $b$-LSC ($a$ is a different point from $b$), a trivial observation is that if $ed(a, b) > d$, one of them can never be a superset of the other and there is no need to perform element-wise set comparison.

 However, in many situations, even though $ed(a, b)$ is smaller than $d$, it is very likely that an $a$-LSC can never cover a $b$-LSC, as Fig.~\ref{fig:prune} (b) shows.  The following will seek a stronger pruning rule in the granularity of LSCs  so that we only need to check  elements of two LSCs  when necessary.

 Assume that there is  a set of points $S$ and there exists a circle $C$ with radius $r$  covering all points in $S$, and we now consider the problem  to decide the location of  $C$. Denote the circle center of $C$ as $C_o$,  for any point $s\in S$, we have $ed(s, C_o) \leq r$. We draw  a circle with radius $r$ centered at each point in $S$, then $C_o$ must lie in the intersection of these circles.  We relax these circles with their minimum bounding rectangles, and  $C_o$ must lie in the intersection area of these rectangles. The intersection  rectangle is trivial to  compute: instead of considering all points in $S$, we only need four values: $x_{max}$, $x_{min}$, $y_{max}$ and $y_{min}$,  which are the maximal and  minimal $x$ coordinates and  $y$ coordinates of points in $S$ respectively. As  Fig.~\ref{fig:prune} (a) shows, there are three points filled with grey  that decide the intersection rectangle.  The dashed rectangle centered at the uppermost or rightmost point decide the  bottom side or left side of intersection rectangle respectively, while the one  centered at the  leftmost and also bottom-most point   decide the right and upper  side of intersection. The rectangle is thus calculated by $CenterRec = \{(x, y)|x\in [x_{max} - r, x_{min} + r], y\in [y_{max} - r, y_{min} + r]\}$. 

\begin{figure}[htbp]
\centering
\includegraphics[width=0.4\textwidth]{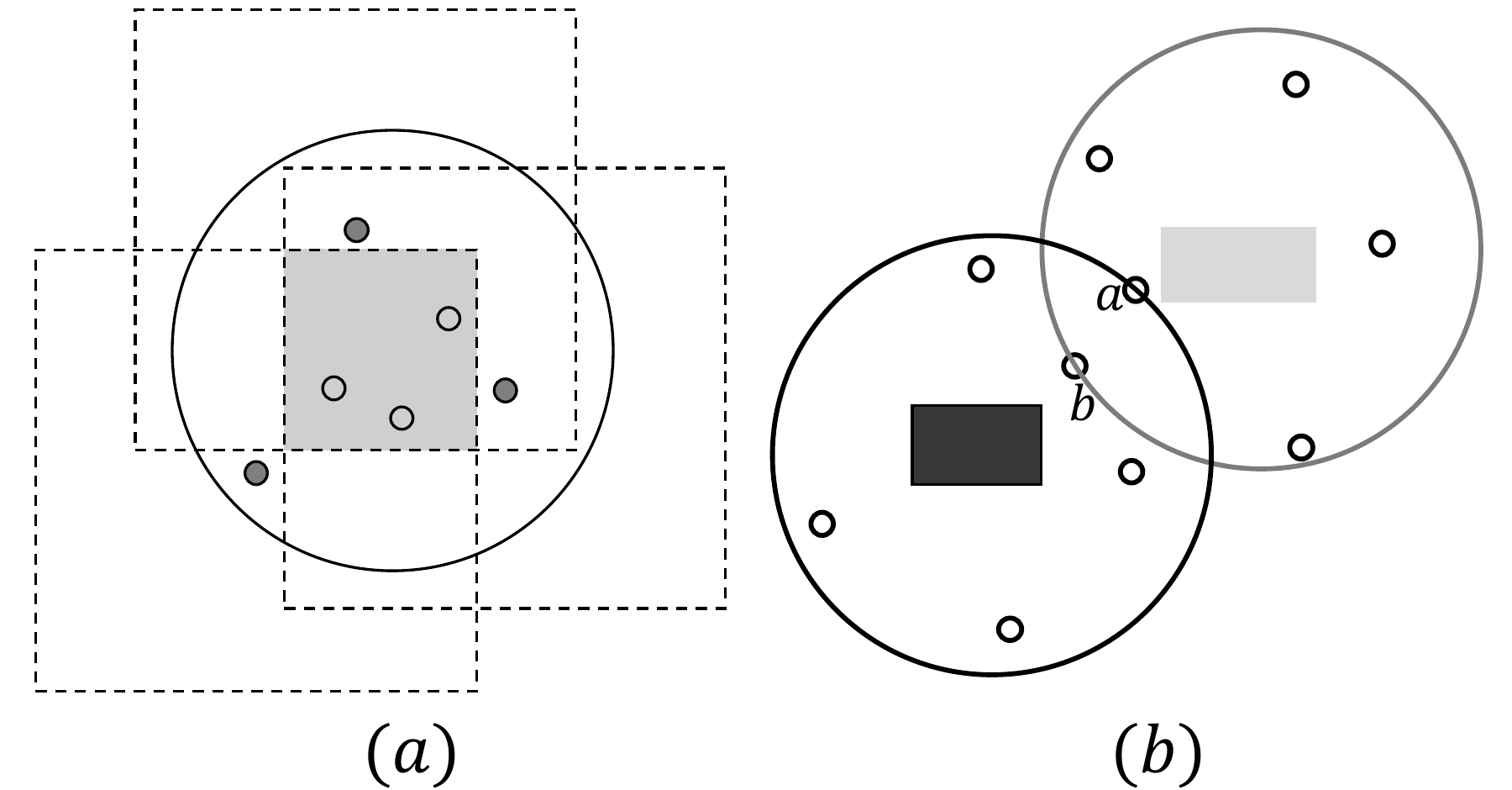}
\setlength{\belowcaptionskip}{-3pt}
\caption{Illustration of pruning rules.}\label{fig:prune}
\end{figure}

For two LSCs with different reference nodes, we consider the necessary condition for a set to cover another. As  Fig.~\ref{fig:prune} (b) show, there are  two bounded  circles with threshold $d$ as diameter, shown as black and grey large circles,  covering an $a$-LSC and $b$-LSC  respectively. For each LSC, we  calculate the rectangle $CenterRec$ as the black and grey shadows show respectively. Since their $CenterRec$s do not intersect with each other, it is not likely to find  a circle with diameter $d$ to cover all points in these two sets, thus 
neither of the two LSCs can cover the other. The following is a stricter pruning rule,

\textbf{Pruning rule II: LSC-wise pruning.}
 Given an LSC $S$,  we only need to do set comparison between $S$ and each of those LSCs whose $CenterRec$ intersect with that of $S$.

\textbf{Implementation.}
By applying these two pruning  rules, we re-implement the function \emph{FindGSC} in  Algorithm~\ref{algo:global_maximal_set}, called \emph{FindGSCPrune}. As Algorithm~\ref{algo:prune1} shows,  for any point $v$, the point-wise pruning rule is first applied. Nearby candidate points $points$ is found by using a range query, and then a set of all  LSCs with reference node in $points$ are gathered for comparison ($cprSets$ in Algorithm~\ref{algo:prune1}). To further reduce set comparisons, for each $v$-LSC $s$, set-wise pruning rule is applied so that we only compare $s$ with  sets in $cprSets$ each of which has a $CenterRec$ intersecting with $s$'s. 

% Each range query using R tree takes time $O(\log n)$. There are around $C_{2d}$ points in $candidatePoints$ and there are $O(C_d)$ $p$-LSCs for any point, so each LSC will be compared with other $O(C_{2d}\times C_{d})$ LSCs, where $C_d = \rho \pi (d/2)^2$.  The complexity for Algorithm \ref{algo:global_maximal_set} with pruning rule I would be $O(nC_d^2 + n\log n + nC_d(C_{2d}\times C_{d}\times C_d)) = O(n\log n + nC_d^3 C_{2d})$.

\begin{algorithm}
\footnotesize
\caption{\emph{Find\_GSC\_With\_PruneRules}}\label{algo:prune1}
\KwIn{Map<Point $v$, a set of $v$-LSCs> where each $v$-LSC has $CenterRec$}
\KwOut{A set of all Global Spatial Clusters $GSC$}
\SetKwFunction{FindGSCPrune}{FindGSCPrune}
\SetKwProg{Fn}{Function}{:}{}
 \Fn{\FindGSCPrune {$LSCMap$}}{
 $GSC = \{ \}$\;
 \For{$v$ in $LSCMap.keys()$}{
 \tcc{apply pruning rule 1}
 $points \gets$  range\_query(v, d)\;
 $cprSets\gets$ the set of all $v'$-LSCs for $v'\in points$\;
 \For{$s\in LSCMap.get(v)$}{
    \For{$s'\in cprSets$}{
    \tcc{apply pruning rule 2}
    \lIf{$s.CenterRec$ intersect with $s'.CenterRec$}{compare $s$ and $s'$}
    }
 }
 }
 \textbf{return} $GSC$\;
 }
\end{algorithm}
\setlength{\textfloatsep}{0pt}

% By using this pruning rule, for any LSC, the number of other LSCs that needed to be compared can be considered as a constant. Range query using R tree takes time $O(\log n)$. Thus the total complexity in the worst case for finding all GSCs would be $O(n^3 + n^2(\log n + cn)) = O(n^3)$.

%% file: exp_results/secs/approx.tex
\section{Approximate Spatial Algorithm} \label{sec:approx}

In last section, we propose powerful pruning rules, though it works in practice, it would still be desirable to pursue a more scalable algorithm for large scale GeoSN. In addition, in the exact algorithm,  only after all LSCs are detected can we decide if an LSC is global. However, in many scenarios, users would expect to get GSCs in a more interactive way, i.e., we should return some GSCs before all LSCs are detected. In this section, we will show that if we loose the spatial constraint, then a much more efficient and interactive algorithm with constant approximation ratio ($\sqrt{2}$) can be designed. 
% The exact algorithm with pruning rules takes $O(n^4)$ time complexity in worst case even with strict pruning rule, it is not efficient in large scale spatial graph, so an approximation efficient algorithm is necessary. 

\subsection{The Basic Intuition}
In Fig.~\ref{fig:ratio}, assume that the small black points consists of a Global Spatial Cluster, then based on the definition, there is a circle, shown as the large black circle,  with diameter $d$ which is the spatial distance threshold to cover this cluster. We relax this circle by its minimum bounding rectangle, shown as the black rectangle  in the figure, and this rectangle must cover all points in that cluster. Similar to the definition of GSC, we give that of Global Approximate Spatial Cluster (GASC) based on rectangles.
\begin{figure}[t]
    \centering
    \includegraphics[width=0.2\textwidth]{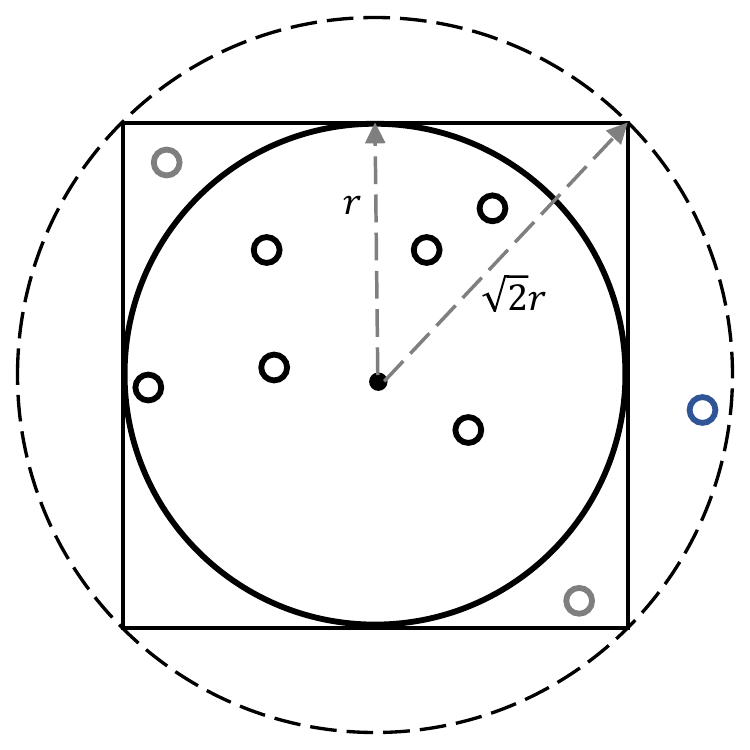}
    %\vspace{-0.8em}
    %\setlength{\belowcaptionskip}{-8pt}
    \caption{Illustration for theorem~\ref{lemma: sandwich}.}
    \label{fig:ratio}
\end{figure}

\begin{definition}[Global Approximate Spatial Cluster] \label{def:globale_maximal_square}
Given a set of points $U\subset V$  and a distance threshold $d$, $U$ is a Global Approximate Spatial Cluster if,
\begin{itemize}
    \item there exists a rectangle $R$ with side length $d$ covering $U$;
    \item there does not exist a rectangle with side length $d$ covering a set of points $U'$  such that  $U\subsetneq U'$.
\end{itemize}
The rectangle $R$ covering $U$ is called a global maximal square.
\end{definition}
 We give a theoretical  bound for using GASCs to replace GSCs,

\begin{theorem}[Sandwich Theorem]\label{lemma: sandwich}
For a distance threshold $d$, denote the set of all GSCs as $\mathcal{U}_d$ and  the set of all GASCs  as $\mathcal{A}_d$. Then we have the following theorem,
\begin{itemize}
\item For any set $U \in \mathcal{U}_d$, $\exists A\in \mathcal{A}_d$ such that $U\subset A$.
\item For any set $A \in \mathcal{A}_d$, $\exists U' \in \mathcal{U}_{\sqrt{2}d}$ such that $A \subset U'$.
\end{itemize}
\end{theorem}
\begin{proof}
Fig.~\ref{fig:ratio} illustrates this lemma. The first property is trivial. For the second property, let the black rectangle denote a global maximal square covering an GASC, then its  minimum bounding circle, denoted as the black dashed circle, with radius $\sqrt{2}d/2$ must cover this GASC.   
\end{proof}

Based on this theorem, the problem of detecting all GSCs can be  approximated by finding all GASCs with approximation ratio $\sqrt{2}$. 
Similar to $v$-bounded circle and $v$-LSC, we give  the definitions of square with $x$-bounded left side (with shorthand as $x$-bounded square) and $x$-Local Approximate Spatial Cluster ($x$-LASC) as follows,
\begin{definition}[Square with $x$-Bounded Left Side] Given a square with side length $d$, it is a square with $x$-bounded left side if the left side of this square passes  node $x$.
\end{definition}
\begin{definition}[Local Approximate  Spatial Cluster] Given an $x$-bounded square $R_x$ and a set of nodes $U_x$ covered by $R_x$, $U_x$ is a $x$-Local Approximate Spatial  Cluster ($x$-LASC) if and only if $U_x\neq \emptyset$ and there does not exist a set of nodes $U'_x\supsetneq U_x$ covered by another $x$-bounded square. Denote the set of all $x$-LASC for a fix $x$ as $\mathcal{U}_x$.
\end{definition}

% We define global maximal square by using a square region with side $d$ to replace circles in the definition of global maximal set. Instead of calculating global maximal set which is covered by a circle with radius $d/2$, we calculate all sets of nodes each of which is covered by a global maximal square.
 
Similar to Lemma~\ref{lemma:correctness}, we have the following 
lemma showing the relationship between GASCs and LASCs.
% The algorithm for finding all GSCs runs in three steps and the first two steps find out all local maximal sets and the last step generates all GSCs from LSCs. Here is a lemma that illustrates the correctness of generating all GSCs from LSCs. 
\begin{lemma}\label{lemma:correctness2}
Given a set of points $V$ in $\mathbb{R}^2$ and a distance threshold $d$, denote the set of all Global Approximate Spatial Clusters as $\mathcal{U}$. It always holds that $\mathcal{U}\subset \cup_{x\in V}\mathcal{U}_x$.
\end{lemma}

%\begin{figure*}[!htb]
%    \begin{minipage}[t]{0.45\linewidth}
%        \centering
%        \includegraphics[width=0.3\textwidth]{figs/scan.%png}
%        \caption{Example of Square Sweep.}
%        \label{scan}
%    \end{minipage}%
%    \hspace{0.1in}
%    \begin{minipage}[t]{0.45\linewidth}
%        \centering
%        \includegraphics[width=0.5\textwidth]{figs/app.p%ng}
%        \caption{$O$-bounded Local Maximal Square.}
%        \label{fig:app-illustrate}
%    \end{minipage}
%\end{figure*} 

Based on Lemma~\ref{lemma:correctness2}, the problem of finding all GASCs can be transformed to finding LASCs as candidates and then generating GASCs from the candidate set.

\subsection{Algorithm}
Algorithm~\ref{algo:global_maximal_square} presents the whole procedure to detect all GASCs interactively by a single scan of all nodes.  Line 2 first sorts points by $x$-coordinates. 
For each point, it generates all LASCs and calls function CheckGlobal to check if each LASC is a global ASC. The following explains the detail of these two procedures. 
\begin{figure}[t]
    \centering
    \includegraphics[width=0.3\textwidth]{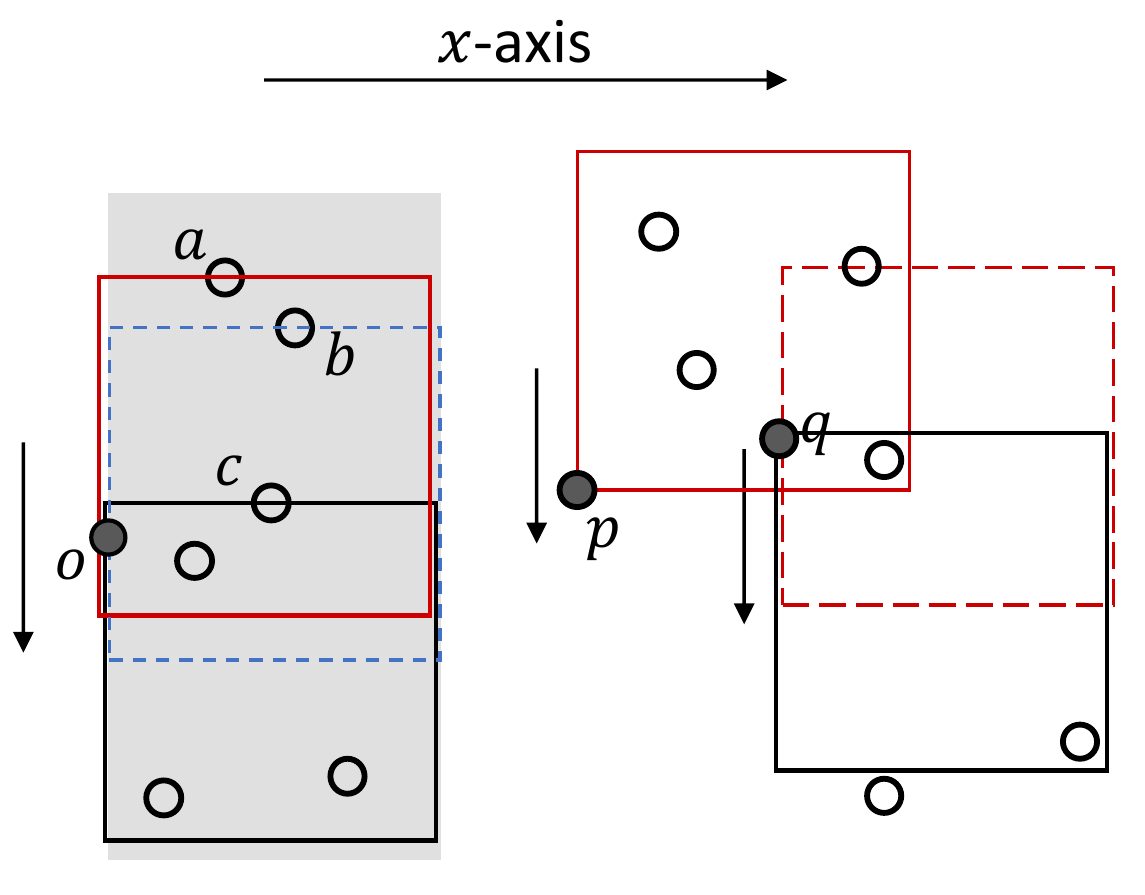}
    \caption{Illustration of Algorithm~\ref{algo:global_maximal_square}}
    \label{fig:GSCquare}
\end{figure}
%\vspace{-1.5em}

\textbf{Detecting LASCs (lines 4-13).} Fig.~\ref{fig:GSCquare} illustrates the steps of finding  LASC.  For point  $O$ with coordinate $(O.x, O.y)$, the possible points that a  $O$-bounded  square can  cover is in the rectangle $\{(x, y)|x\in [O.x, O.x+d], y\in[O.y-d, O.y+d]\}$ as the grey shadow  rectangle  shows. Then we generate all $O$-LASCs by moving a $d\times d$ rectangle downwards in the shadow. The points in the grey shadow are sorted by $y$ coordinates, and then  let the upper side of an $O$-bounded square pass the  points  one by one. Let $start$ keeps tracking of the first point that has not been covered and $CS$ denote currently covered. Initially, let the upper side of rectangle passes the first point $a$ (line 6). Get all points covered by this rectangle (lines 12 - 13) and this should be a  LASC since when the rectangle moves downwards, it can not cover $a$ anymore. Move it downwards so that its upper side passes the next point,  and there can be  three possible situations: if   the last point  in the grey shadow has been  enclosed by the previous rectangle, stop moving and terminate (line 9); if the new rectangle does not enclose any new node, then ignore it (line 10); if the new rectangle encloses more points than the previous one, then get all nodes covered by this square and it is a LASC (lines 11-13). For example, in Fig.~\ref{fig:GSCquare}, when the rectangle moves to pass the second point $b$, no new points is covered in the blue dashed rectangle and thus it is skipped,  while when it passes $c$, new points are included and all  points enclosed forms a LASC. Since the last point has been covered, it terminates. 

\textbf{Finding GASCs (lines 19-25).}  Once an LASC is found, it is easy to check if this is a GASC. For example, in Fig.~\ref{fig:GSCquare}, there is a $q$-LASC covered by the  red dashed rectangle. To check if this is a GASC, we  only need to compare it with $p$-LASCs where $p.x< q.x$ which have already been detected,  since any $q'$-LASC where $q'.x > q.x$ can not contain point $q$. Only at most three points in the dashed rectangle  needed to be considered. The three points are: point with maximum $x$ coordinate and points with minimum and maximum $y$ coordinates. If these three points  are already in a previous LASC,  then all points in the rectangle are in it, thus this LASC will be discarded.  Otherwise, this is an GASC.
Function CheckGlobal of Algorithm~\ref{algo:global_maximal_square} shows this  process where $NodeGASC$ records for each point a set of  all GASCs currently found  that enclose it. By determining if the three GASCs sets for these three special points have intersection, we can check if the LASC is a global one.   
%Fig.\ref{scan} illustrates this idea. Suppose that red rectangle and blue rectangle cover a $O$-LASC and $d$-LASC respectively where $O.x < d.x$, for the points covered in  red square, they can not be fully covered by a $d$-bound square because at least $O$ can not be  covered. Thus, we only need to check a $O$-bounded local maximal square with LASCs found so far to determine if this is a GASC and there is no need to check it with future LASCs. 
\setlength{\textfloatsep}{0pt}
\begin{algorithm}
\footnotesize
\caption{\emph{Find\_Approximate\_Spatial\_Clusters}}\label{algo:global_maximal_square}
\KwIn{A set of nodes $V$ of geo-social network $G$, distance threshold $d$, social constraint $k$}
\KwOut{A set of all GASCs $GASC$}
\SetKwFunction{CheckGlobal}{CheckGlobal}
\SetKwProg{Fn}{Function}{:}{}
$GASC \gets \{ \}$, $NodeGASC \gets Map< >$, $GASCLabel \gets 0$ \;
Sort $V$ in the increasing order of $x$ coordinate, and denote the sorted list as $C= \{p_1, p_2, \cdots, p_N\}$  where $p_i$ has $p_i.x$ and $p_i.y$\;
\For{$p_i$ in $C$} {
     \tcc{do a range query}
     $L \gets \{ p_j\in V| p_j.x\in[p_i.x, p_i.x+d], p_j.y\in[p_j.y - d, p_j.y + d] \}$\;
     \tcc{sorted in the increasing order of $p_j.y$}
    $L\gets L.sort(key = p.y)$\;
	$start \gets 1$; $CS \gets  \{start\}$\;

	\For{$j$ from $1$ to $L.$length}{
	    $CS.pop(0)$\;
	    \lIf{$start > L.length$}{
	    break}
	    \lElseIf{$p_{start}.y - p_j.y > d$}{
	    continue}
	    %\lIf {$p_{start}.y - p_j.y > d$}{continue}
	    \Else{
	    \While {$start\leq L.length$ and  $p_{start}.y - p_j.y \leq d$}
	    {$CS.add(p_{start})$; $start \gets start+1$\;}
	    \If {\CheckGlobal{$NodeGASC$, $CS$} == True}{
	    $GASC.add(CS)$\;
	    \lFor{$p$ in $CS$}{$NodeGASC[p].add(GASCLabel)$}}
	    $GASCLabel ++$ ;
	    }
	}
}
\textbf{return} $GASC$\;
\Fn{\CheckGlobal{$NodeGASC $, $CS$}}{
$x, y \gets $ the first, last node in $CS$\;
\lIf{$(S \gets NodeGASC[x] \cap NodeGASC[y]) = \emptyset$}{\textbf{return} True}
\Else{
$z  \gets$ node with the maximal $x$ coordinate in $CS$\;
\lIf{$S\cap NodeGASC[z]) = \emptyset$}{\textbf{return} True}
\lElse{\textbf{return} False}
}
}
\end{algorithm}
\setlength{\textfloatsep}{0pt}

 \textbf{Complexity Analysis} The average number of points in a $d\times 2d$ rectangle is $C = 2\rho d^2$. For each point $v\in V$, lines 7 to 13 take time $O(C)$ to compute all $v$-LASCs. For function CheckGlobal,  the dominate step is set intersections. Suppose there are $O(D)$ GASCs that may contain a  point $x$, i.e., the size of $NodeGASC[x]$ is  $O(D)$, then  conducting a set intersection operation would take $O(D)$, then the total time complexity  is $O(n\log n + n\times (C \times D)) = O(n\log n + nCD)$. In the worst case, $D = O(C^2)$, however, since  $NodeGASC[x]$ records only GASCs currently found instead of all LASCs that contain $x$, $D$ is practically very small.

%% file: exp_results/secs/experiment.tex
\section{Experimental Studies}\label{sec:exp}
Our experiments contain three  parts: we first test and compare the spatial algorithms which find out all spatial clusters, then test the whole MCC framework to get all maximal co-located communities, and finally we conduct case studies to compare our results with two state-of-the-art researches. All of our algorithms are implemented by Java using JDK 11 and tested on an Ubuntu server with Intel(R) Xeon(R) CPU X5675 @ 3.07GHz and 64 GB memory.

\begin{figure*}
\centering 
\subfigure[time vs. $N$ (Uniform)]{\label{fig:exp_s_uniform}
    \begin{minipage}[b]{0.192\textwidth}
        \centering 
        \includegraphics[width=\textwidth]{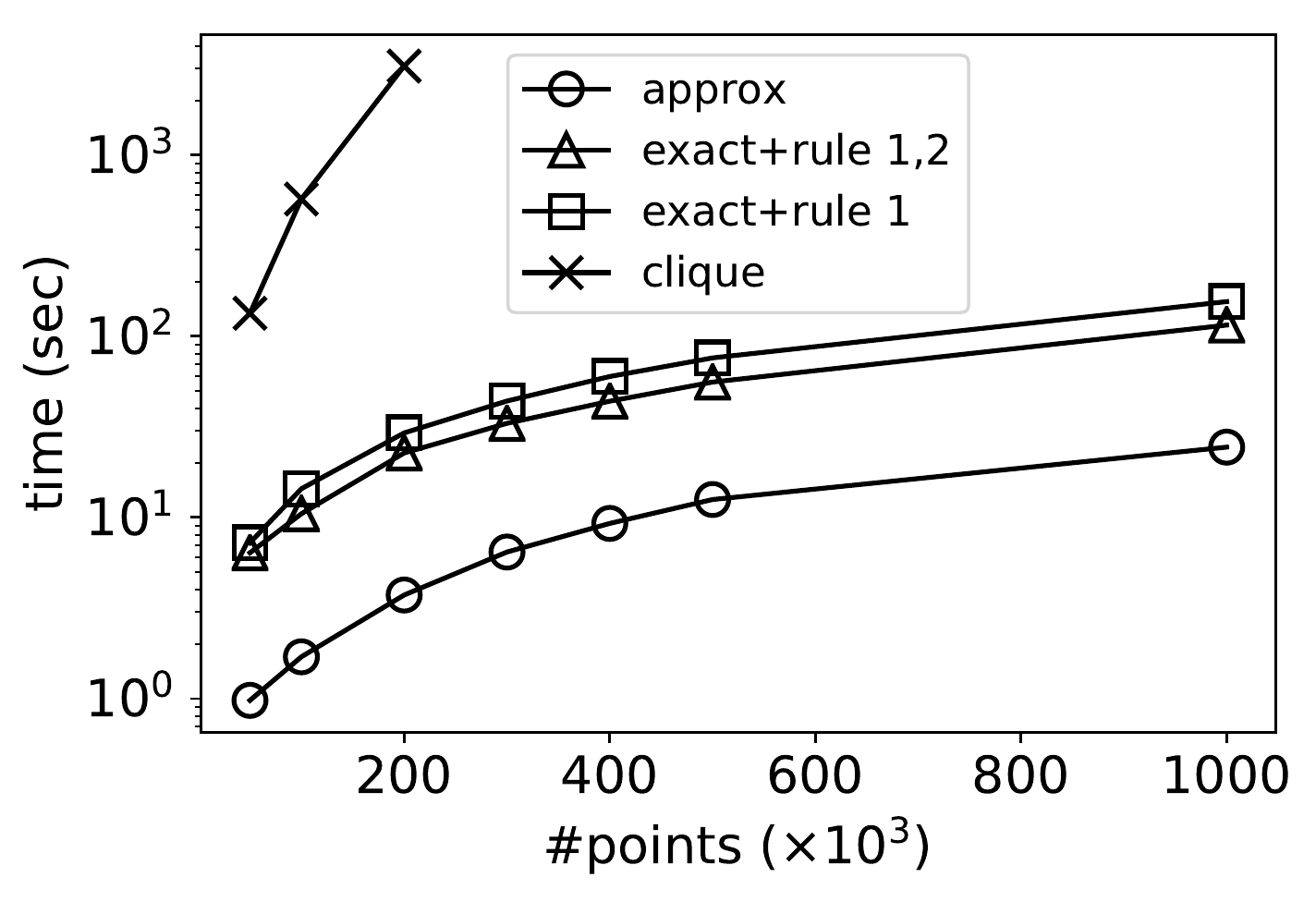}
    \end{minipage}}
\subfigure[time vs. $N$ (Gaussian)]{\label{fig:exp_s_gaussian}
    \begin{minipage}[b]{0.192\textwidth}
        \centering 
        \includegraphics[width=\textwidth]{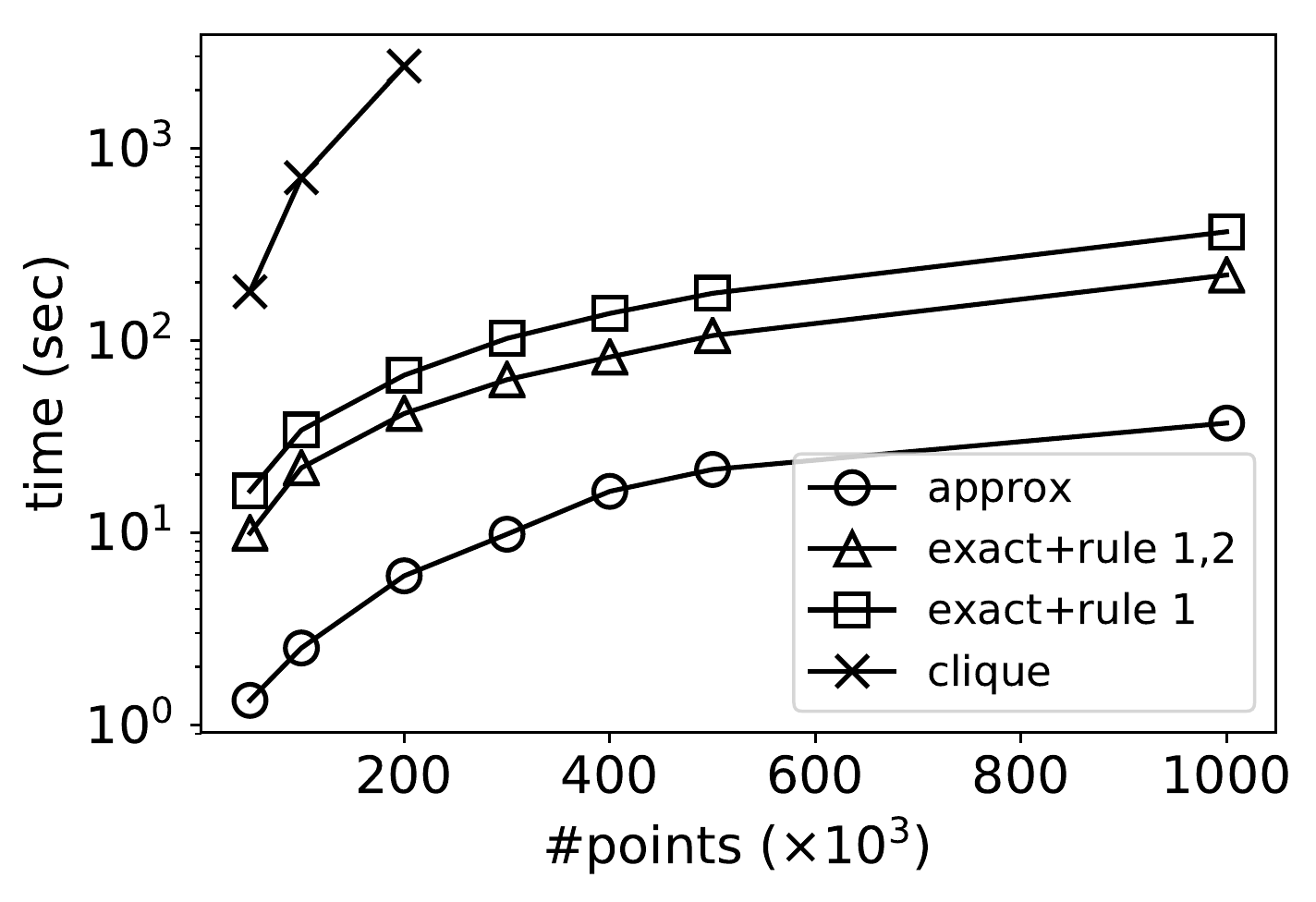}
    \end{minipage}}
\subfigure[time vs. $ratio$ (Brightkite)]{\label{fig:exp_s_brightkite}
    \begin{minipage}[b]{0.192\textwidth}
        \centering 
        \includegraphics[width=\textwidth]{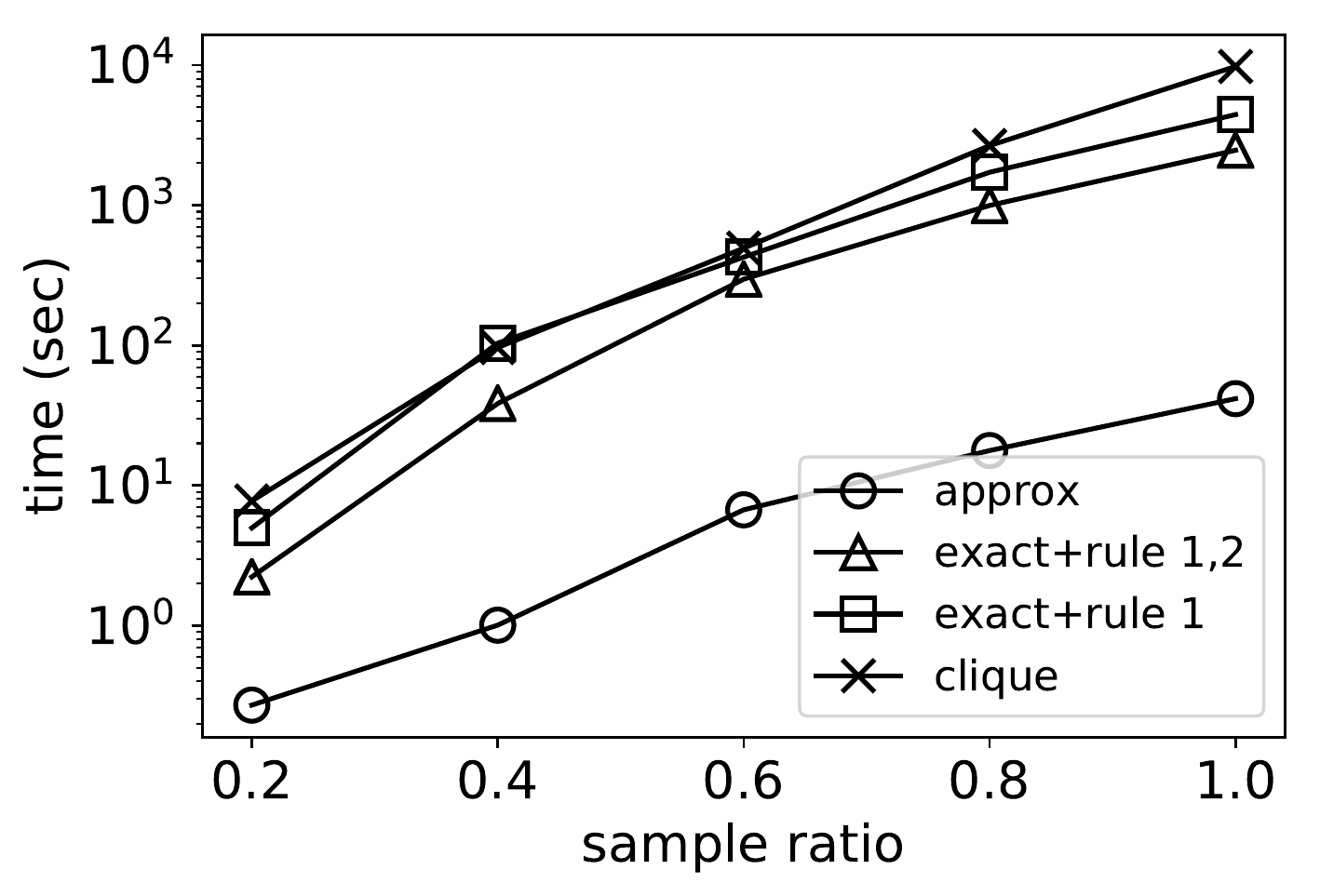}
    \end{minipage}}
\subfigure[time vs. $ratio$ (Gowalla)]{\label{fig:exp_s_gowalla}
    \begin{minipage}[b]{0.192\textwidth}
        \centering 
        \includegraphics[width=\textwidth]{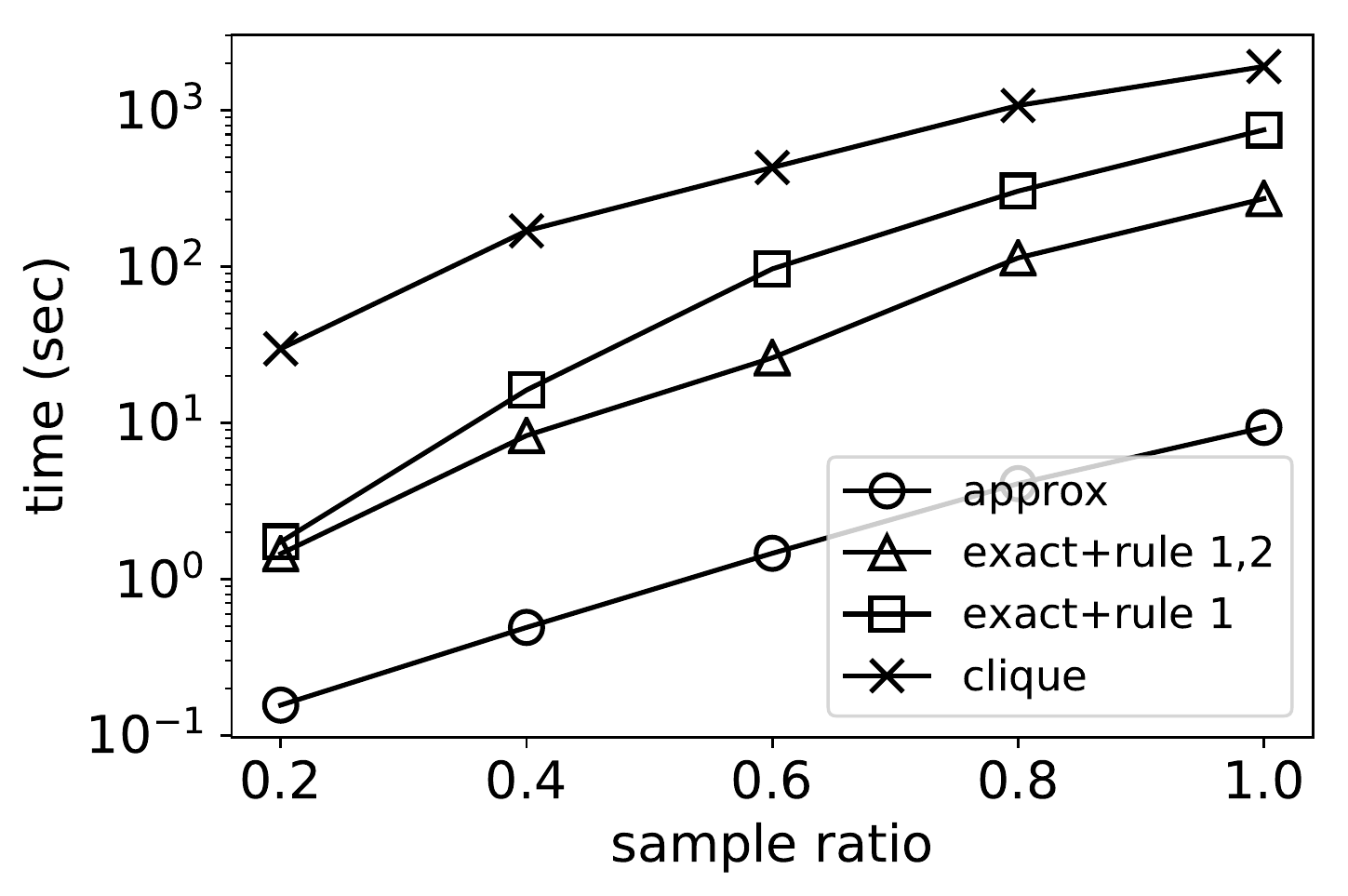}
    \end{minipage}}
\subfigure[time vs. $ratio$ (Weibo)]{\label{fig:exp_s_weibo}
    \begin{minipage}[b]{0.192\textwidth}
        \centering 
        \includegraphics[width=\textwidth]{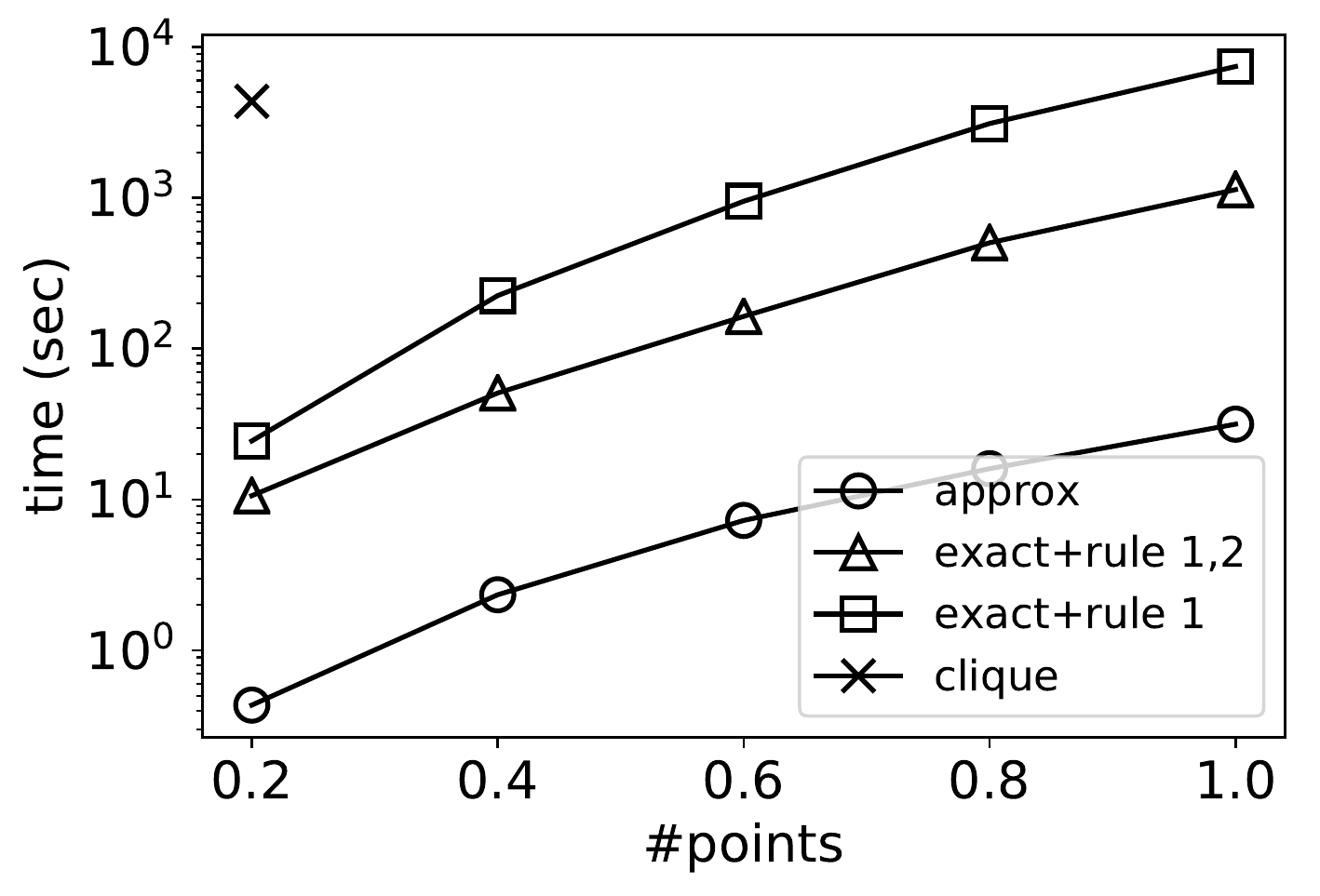}
    \end{minipage}}
% \caption{Experiment results of scalability on synthetic and real-world datasets.}\label{fig:exp_scalability}
% \end{figure*}

% \begin{figure*}
% \centering
\subfigure[time vs. $d$ (Uniform)]{\label{fig:exp_d_uniform}
    \begin{minipage}[b]{0.192\textwidth}
        \centering 
        \includegraphics[width=\textwidth]{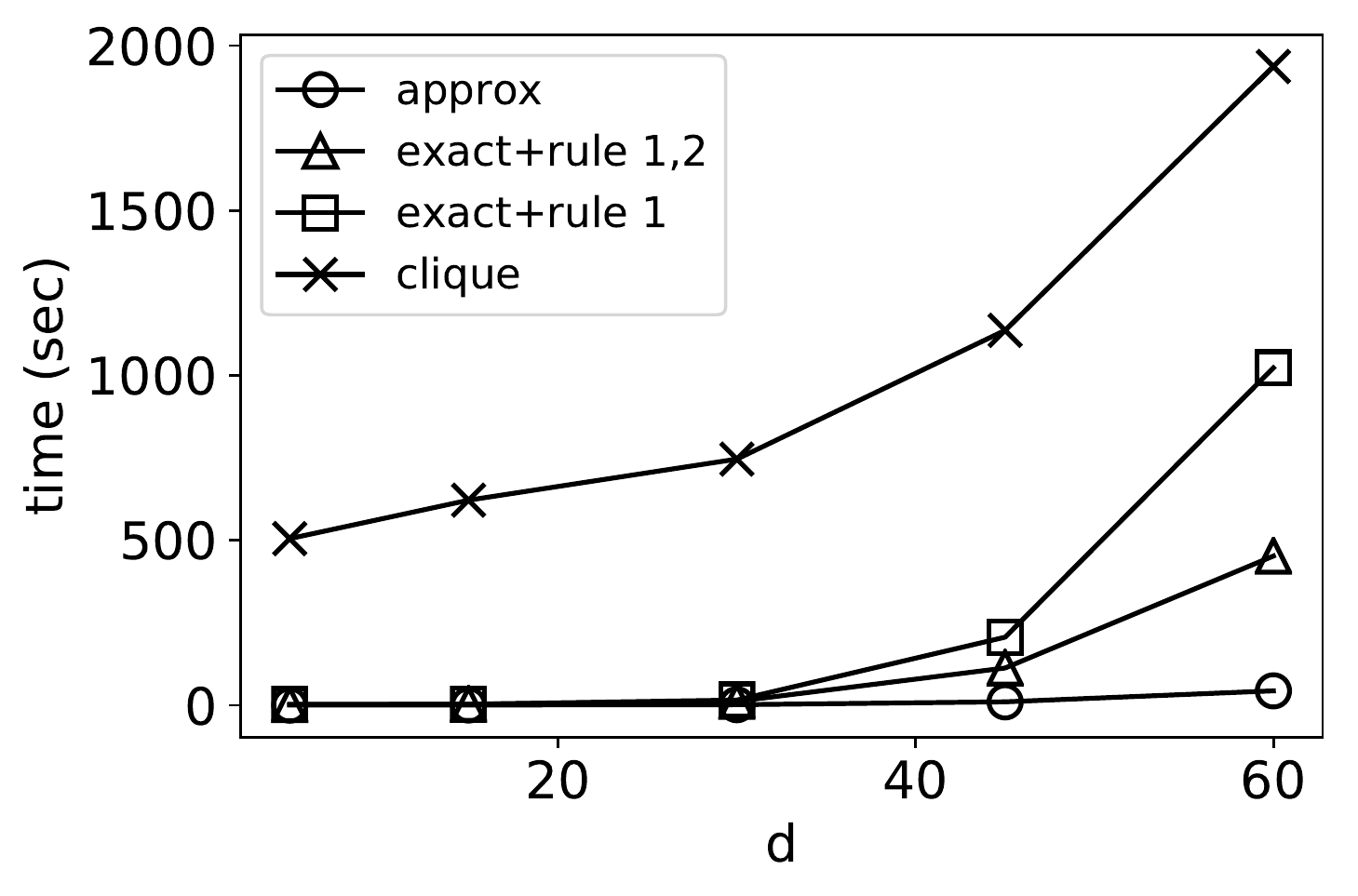}
    \end{minipage}}
\subfigure[time vs. $d$ (Gaussian)]{\label{fig:exp_d_gaussian}
    \begin{minipage}[b]{0.192\textwidth}
        \centering 
        \includegraphics[width=\textwidth]{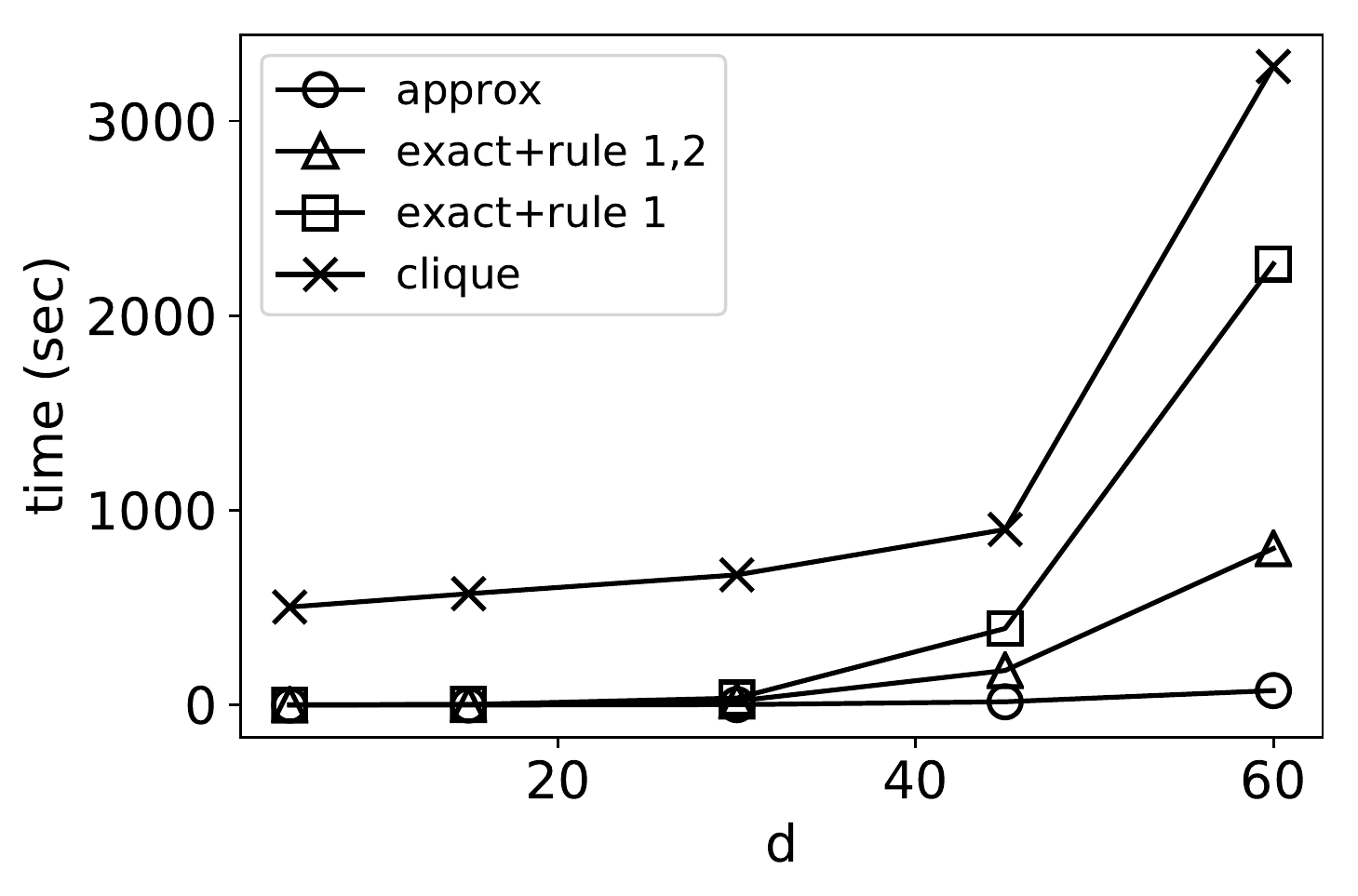}
    \end{minipage}}
\subfigure[time vs. $d$ (Brightkite)]{\label{fig:exp_d_brightkite}
    \begin{minipage}[b]{0.192\textwidth}
        \centering 
        \includegraphics[width=\textwidth]{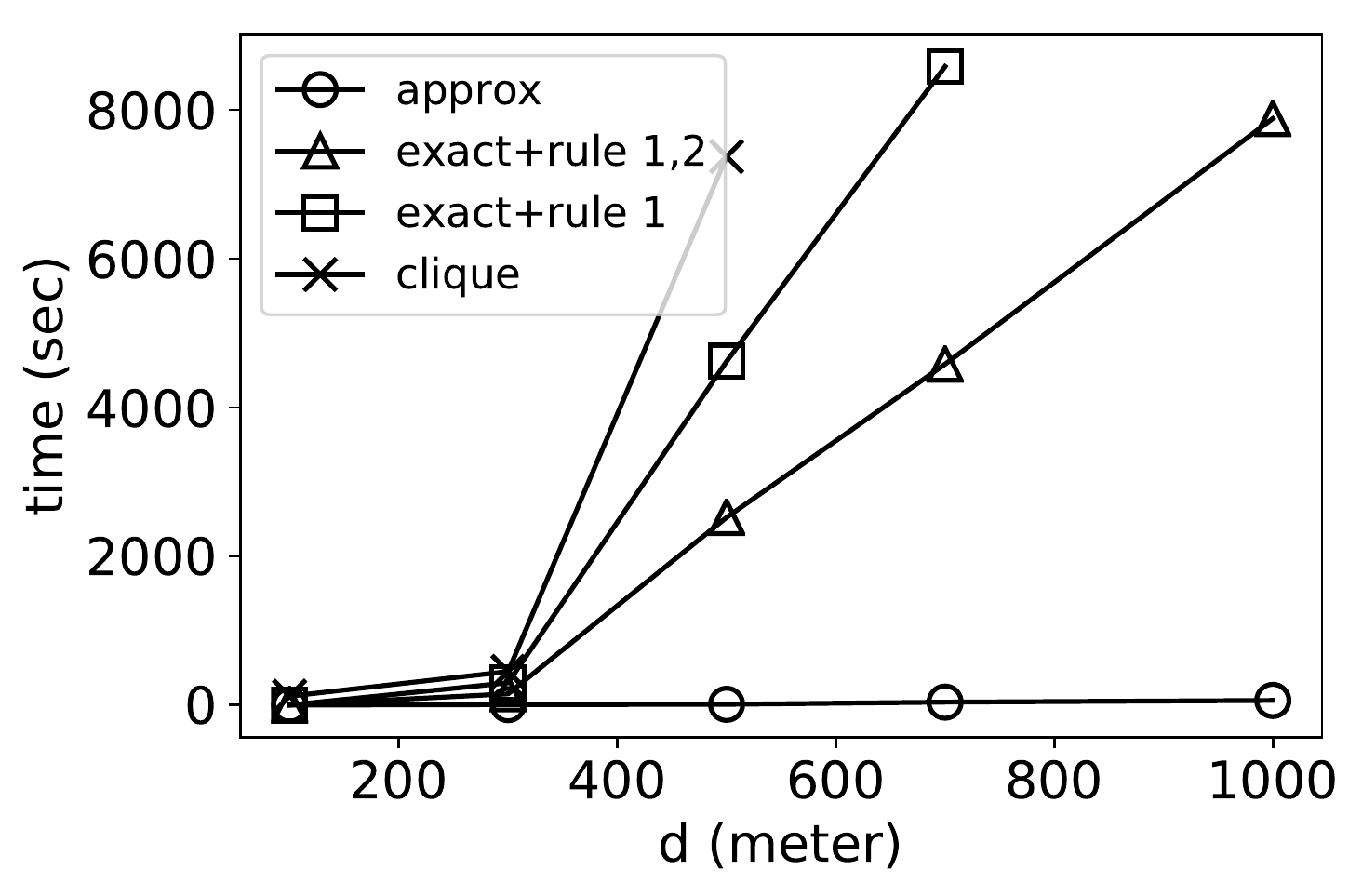}
    \end{minipage}}
\subfigure[time vs. $d$ (Gowalla)]{\label{fig:exp_d_gowalla}
    \begin{minipage}[b]{0.192\textwidth}
        \centering 
        \includegraphics[width=\textwidth]{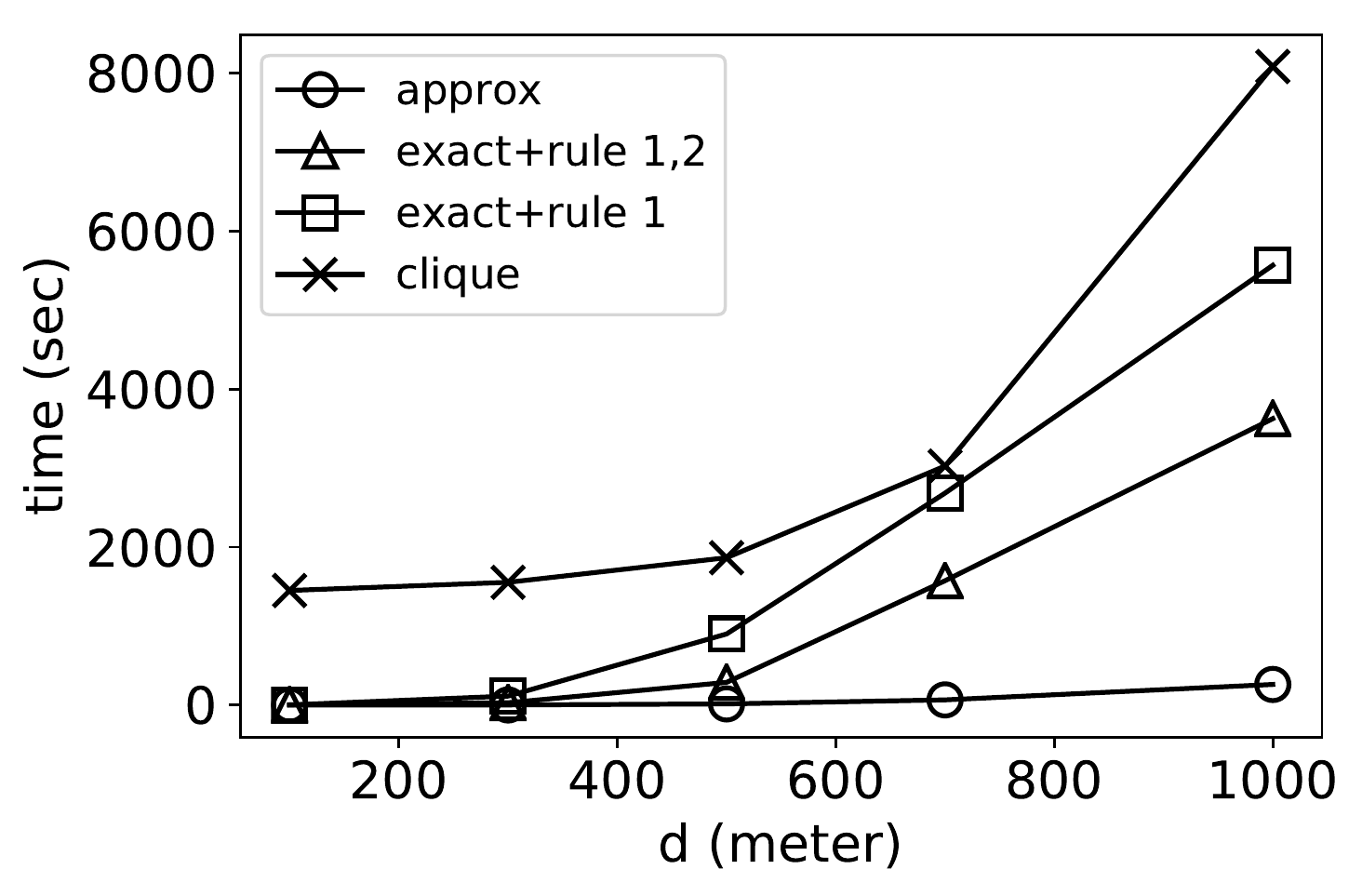}
    \end{minipage}}
\subfigure[time vs. $d$ (Weibo)]{\label{fig:exp_d_weibo}
    \begin{minipage}[b]{0.192\textwidth}
        \centering 
        \includegraphics[width=\textwidth]{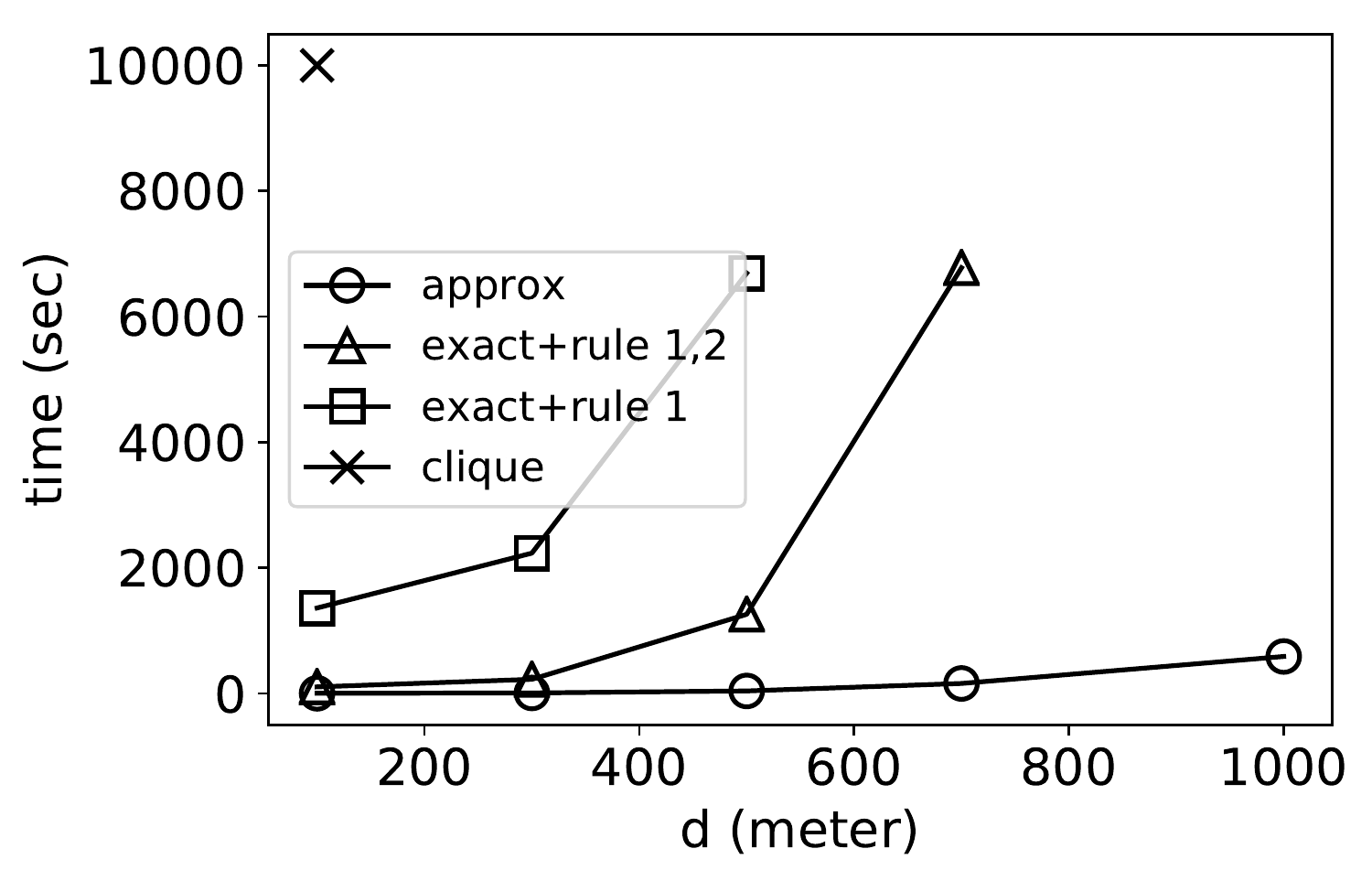}
    \end{minipage}}
%\vspace{-1.5ex}
\setlength{\belowcaptionskip}{-10pt}
\caption{Experiment results of the influence of scalability and $d$ on synthetic and real-world datasets.}\label{fig:exp_spatial_all}
\end{figure*}

\subsection{Spatial Algorithm Evaluation}
In this part, we test the performance of our spatial  algorithms and compare with the state-of-the-art algorithm from \cite{zhang2017engagement, chen2018maximum}.

\noindent \underline{\textbf{Algorithms}} We test four algorithms as shown in Table.~\ref{tab:algo}. The clique-based algorithm from \cite{chen2018maximum} does not solve exactly the same problem as ours, however, it can be easily adapted to detect all spatial clusters  by enumerating all cliques of a spatial graph.
%\vspace{-2em}
\begin{table}[htbp]\caption{List of implemented algorithms.}\label{tab:algo}

\footnotesize
\begin{tabular}{|c|c|}
\hline
\textbf{Name} & \textbf{Algorithm}  \\ \hline\hline
$\mathsf{clique}$ & adapt clique-based algorithm \cite{chen2018maximum} to get all spatial clusters           \\ \hline
$\mathsf{exact + rule1}$         & algorithm~\ref{algo:global_maximal_set} with pruning rule 1      \\ \hline
$\mathsf{exact + rule 1, 2}$ & algorithm~\ref{algo:global_maximal_set} with pruning rule 1 and 2 
\\ \hline
$\mathsf{approx}$ & approximation algorithm  Algorithm~\ref{algo:global_maximal_square}
\\ \hline
\end{tabular}
\end{table}
%\vspace{-1em}

\noindent \underline{\textbf{Dataset}} The experiments are conducted on three real-world datasets and two synthetic datasets. Table.~\ref{tab:statics} presents the statistics of the spatial part of three real geo-social networks.
% Size is the number of users in the dataset who has at least one check in location. 
\#Neighbors is defined as the number of people within 500 meters from a specific user and we calculate the average and  maximum number of \#Neighbors. The locality level is defined as the ratio between max. \#Neighbors and avg. \#Neighbors. For example, for the weibo dataset, the max. \#Neighbors is high while avg. \#Neighbors is low, so it has relatively high spatial  locality. 
%\vspace{-1.5em}
\begin{table}[htbp]\caption{Statistics of real-world spatial datasets.}\label{tab:statics}
\footnotesize
\begin{tabular}{|c|c|c|c|c|}
\hline
\textbf{Dataset} & \textbf{Size (1K)} & \textbf{\begin{tabular}[c]{@{}c@{}}Max. \\ \#Neighbors\end{tabular}} & \textbf{\begin{tabular}[c]{@{}c@{}}Avg. \\ \#Neighbors\end{tabular}} & \textbf{Locality} \\ \hline\hline
Brightkite \cite{cho2011friendship}     & 51        & 1342                                                                 & 55.67                                                                & medium              \\ \hline
Gowalla  \cite{cho2011friendship}        & 107       & 536                                                                  & 15.38                                                                & low               \\ \hline
Weibo\cite{li2014efficient}         & 1,020     & 976                                                                  & 15.85                                                                & high            \\ \hline
\end{tabular}
\end{table}

\noindent \underline{\textbf{Parameter Settings}} Table.~\ref{tab:parameter} shows the parameter settings for both synthetic and  real datasets. For the synthetic datasets, there are three parameters: the number of points $N$, density and distance threshold $d$. For real datasets, we consider two  
parameters:  the percentage of users sampled from the original datasets $ratio$ and distance threshold $d$. At each time, we vary one parameter while other parameters are set to their underlined default values.

% \usepackage{multirow}
% \usepackage{multirow}
\begin{table}[htbp]\caption{Table of parameter setting.}\label{tab:parameter}

\footnotesize
\begin{tabular}{|c|c|c|}
\hline
\textbf{Category}               & \textbf{Parameter} & \textbf{Values}                                 \\ \hline\hline
\multirow{3}{*}{synthetic} & $N$                & $[50, 100, 200, \underline{300}, 400, 500, 1000] \times 1K$ \\ \cline{2-3} 
                                & $density$          & $[0.001, 0.002, 0.004, \underline{0.008}, 0.012, 0.02]$     \\ \cline{2-3} 
                                & $d$                & $[5, 15, \underline{30}, 45, 60]$                           \\ \hline
\multirow{2}{*}{real}      & $ratio$            & $[0.2, 0.4, 0.6, 0.8, \underline{1.0}]$                     \\ \cline{2-3} 
                                & $d$                & $[100, 300, \underline{500}, 700, 1000]$                    \\ \hline
\end{tabular}
\end{table}

\subsubsection{Scalability}
To test the scalability of our algorithms, we vary $N$ for synthetic datasets and $ratio$ for real datasets. The results are shown in Fig.~\ref{fig:exp_s_uniform}-\ref{fig:exp_s_weibo}. The clique-based algorithm increases exponentially as the number of points increases on all datasets, which demonstrates the NP-hardness nature of the clique enumeration problem. Our exact algorithm significantly outperforms clique-based algorithm: 1) on synthetic datasets, it
outperforms clique by one to two orders of magnitudes and clique can not terminate in 8,000 s for more than 500K data points while exact can return results in 100 seconds. 2) on real datasets, our algorithm outperforms clique especially for  large-sized Weibo dataset where  clique can not return results in 8,000 s when sampling  40\% data points. Our exact and approximation algorithms show strong scalability  on synthetic datasets, since they show near-linear increase when the number of points increases. Notably, on three real datasets, the increase is faster than that on synthetic data since when the $ratio$ becomes larger, not only the number of data points, but also the density increase.  % Pruning rule 2 shows effectiveness on all datasets: 1) From % Fig.~\ref{fig:exp_s_uniform} to \ref{fig:exp_s_gowalla}, % exact+rule 1, 2 decreases time by around a half than  % exact+rule1. 2) For Weibo dataset, as % Fig.~\ref{fig:exp_s_weibo} shows, pruning rule 2 decreases % time by orders.  

\subsubsection{Effect of $d$} 
Fig.~\ref{fig:exp_d_uniform}-\ref{fig:exp_d_weibo} present the execution time by varying the distance threshold $d$. For clique, when $d$ increases, the execution time increases dramatically, e.g., it can not return results when $d = 700$ meters for Brightkite and weibo datasets. For exact algorithms, when implemented with both two pruning rules, the execution time is much less than that of using only one pruning rule and it becomes more obvious when $d$ increases. The execution time of approx does not show obvious change w.r.t. $d$ comparing to other algorithms. For the weibo dataset, when $d$ is set as 700 or 1000 meters, approx still return results in short time while other three algorithms cannot terminate within 8,000 s.
 
We briefly give the reason here. For clique, when $d$ increases, the virtual spatial neighborhood network would be more complex and thus enumerating all maximal cliques would be much more time-consuming. For the exact algorithm, as $d$ increases, the number of LSCs and the number of points in each LSC increase, and accordingly the time spent on set comparisons for LSCs would be a major bottleneck. Recall that the time complexity for the exact algorithm has $T \propto n^2d^5$.
Pruning rule 1 decreases the times of set comparisons by excluding comparisons between two LSCs with reference nodes distance larger than $d$. When $d$ increases, the percentage  of set comparisons pruned by this rule would decrease and thus it loses the pruning power. However, for pruning rule 2, it is still very effective when $d$ grows since it is a set-wise pruning method instead of point-wise. 

\subsubsection{Effect of data density} 
Fig.\ref{fig:exp_density} shows the execution time w.r.t. different densities of synthetic datasets. As  density increases, the execution time of both clique and exact+rule1 increases quickly, however, when implemented with both pruning rules, the exact algorithm grows much more slowly. The effectiveness of pruning rule II becomes more obvious with the increase of density. For Gaussian distributed datasets, which have higher locality than uniform data, the pruning rule II reduces more than 50\% execution time than exact+rule1 when density is set as 0.02. Density affects exact algorithms due to the same reason as $d$ does, both of them increase the number and set size of LSCs, which makes LSCs comparisons more costly. The execution time of approximation algorithm increases very slowly since there are  only at most three set comparisons needed to be conducted for each local approximate spatial cluster even though density changes. 

\begin{figure}[th]
    \centering
    \subfigure[time vs. $density$ (Uniform)]{\label{fig:exp_density_uniform}
    \begin{minipage}[b]{0.23\textwidth}
        \centering 
        \includegraphics[width=\textwidth]{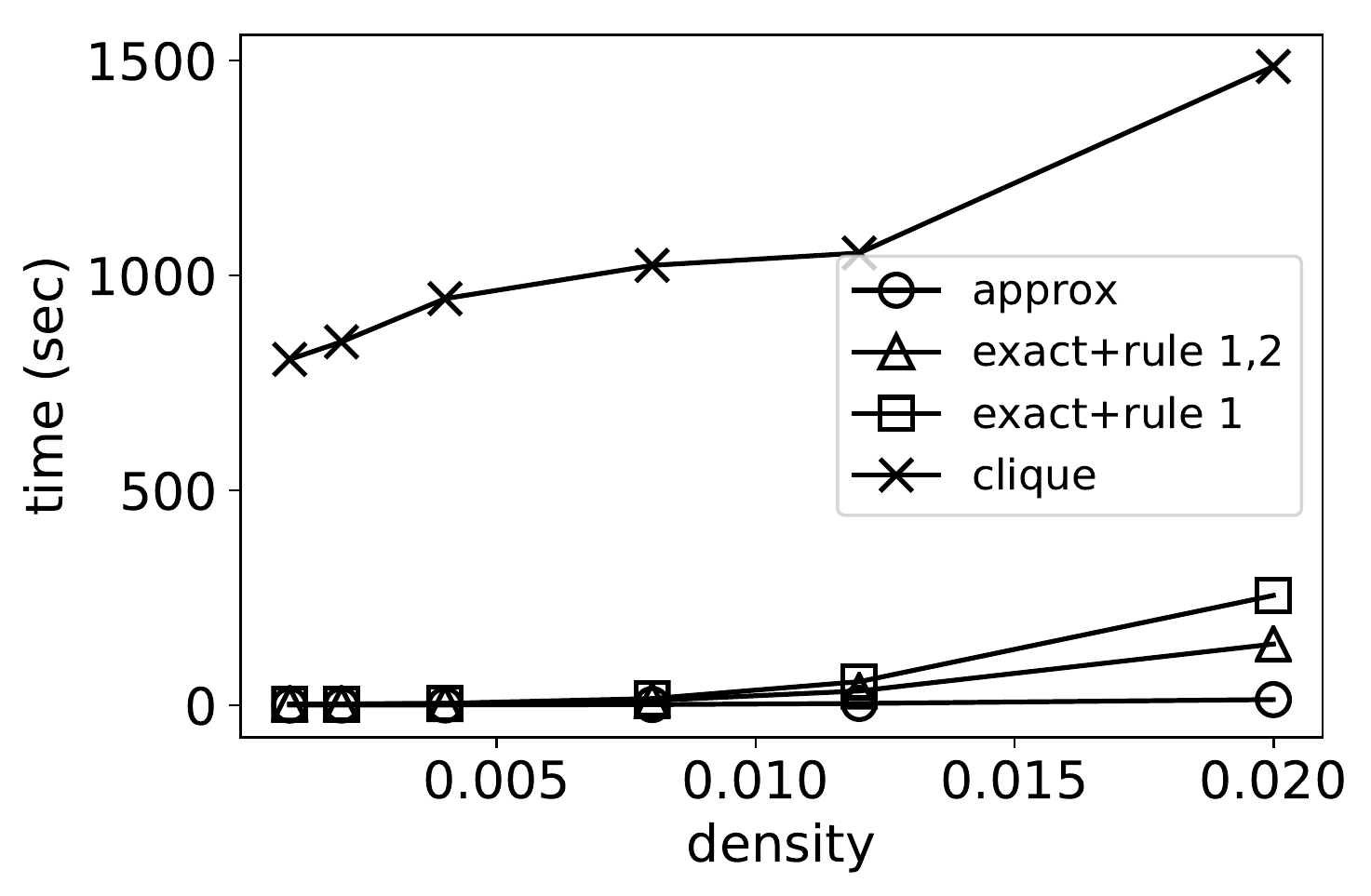}
    \end{minipage}}
    \subfigure[time vs. $density$ (Gaussian)]{\label{fig:exp_density_gaussian}
    \begin{minipage}[b]{0.23\textwidth}
        \centering 
        \includegraphics[width=\textwidth]{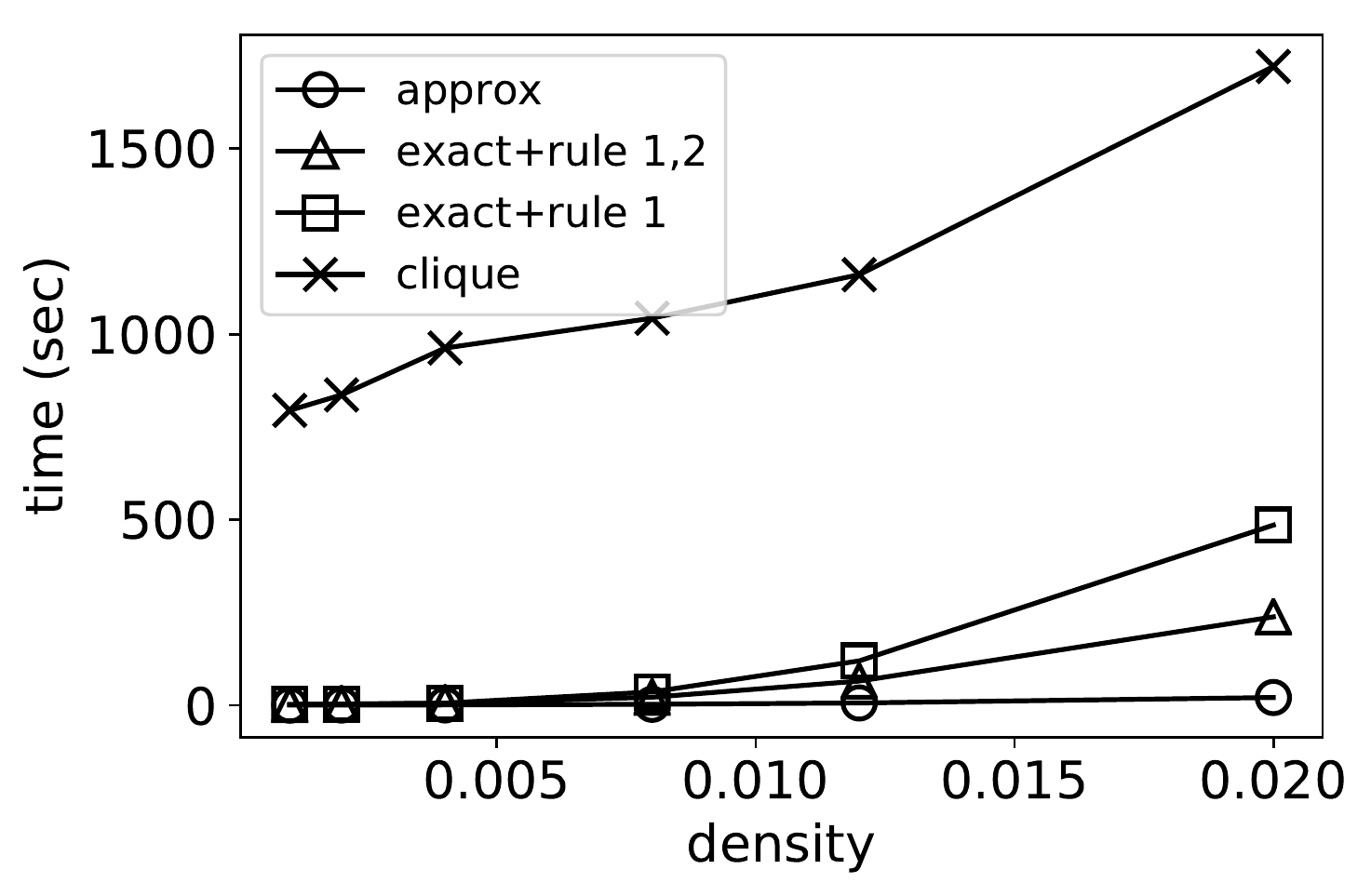}
    \end{minipage}}
    %\vspace{-1.5ex}
    \setlength{\belowcaptionskip}{-2.5pt}
    \caption{Experiment results of the influence of $density$.} \label{fig:exp_density}
\end{figure}
\vspace{-1em}
\begin{figure}[htbp]
    \centering
    \subfigure[when $d$ varies (Gowalla)]{
    \begin{minipage}[b]{0.23\textwidth}
        \centering 
        \includegraphics[width=\textwidth]{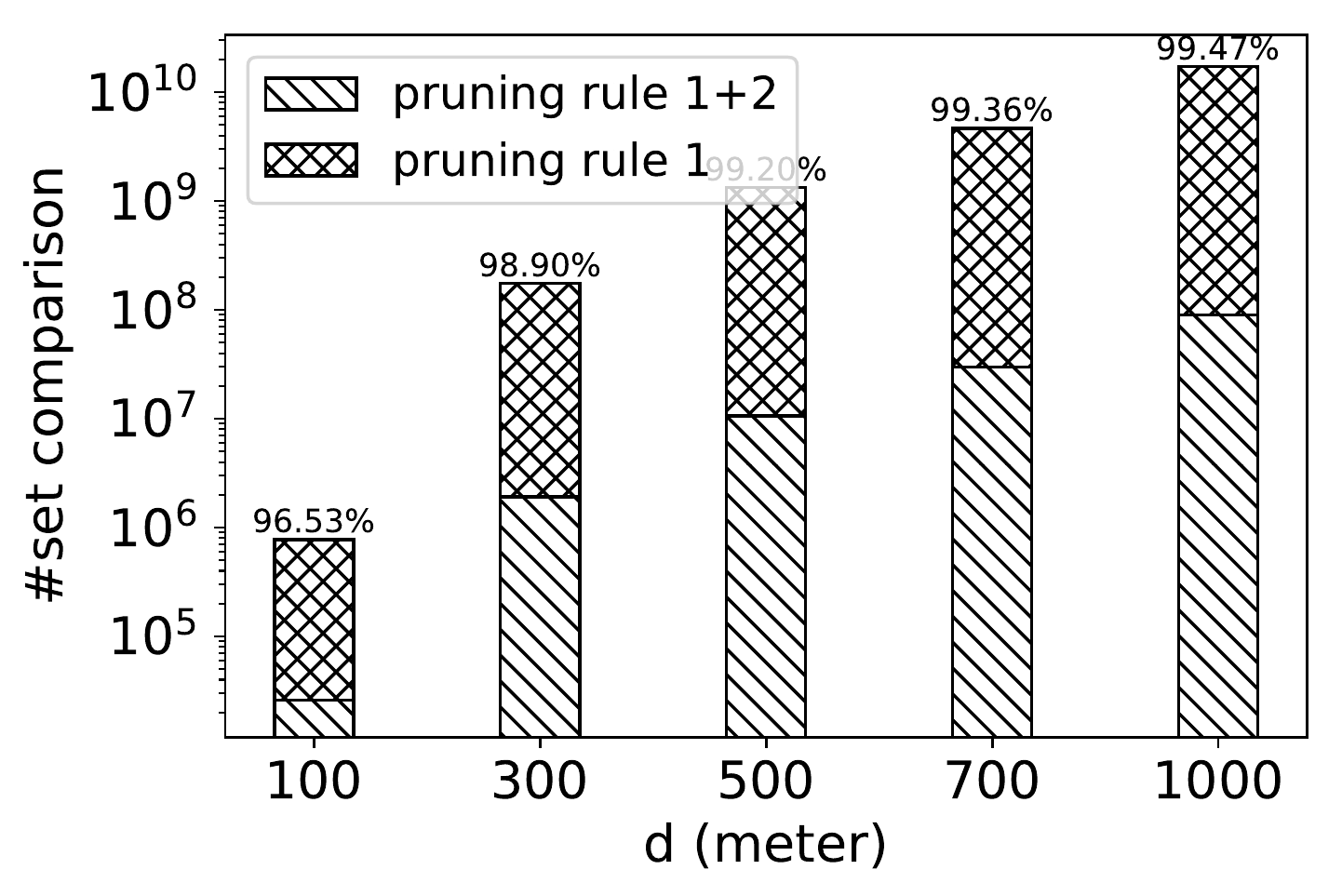}
    \end{minipage}}
    \subfigure[when $density$ varies (Uniform)]{
    \begin{minipage}[b]{0.23\textwidth}
        \centering 
        \includegraphics[width=\textwidth]{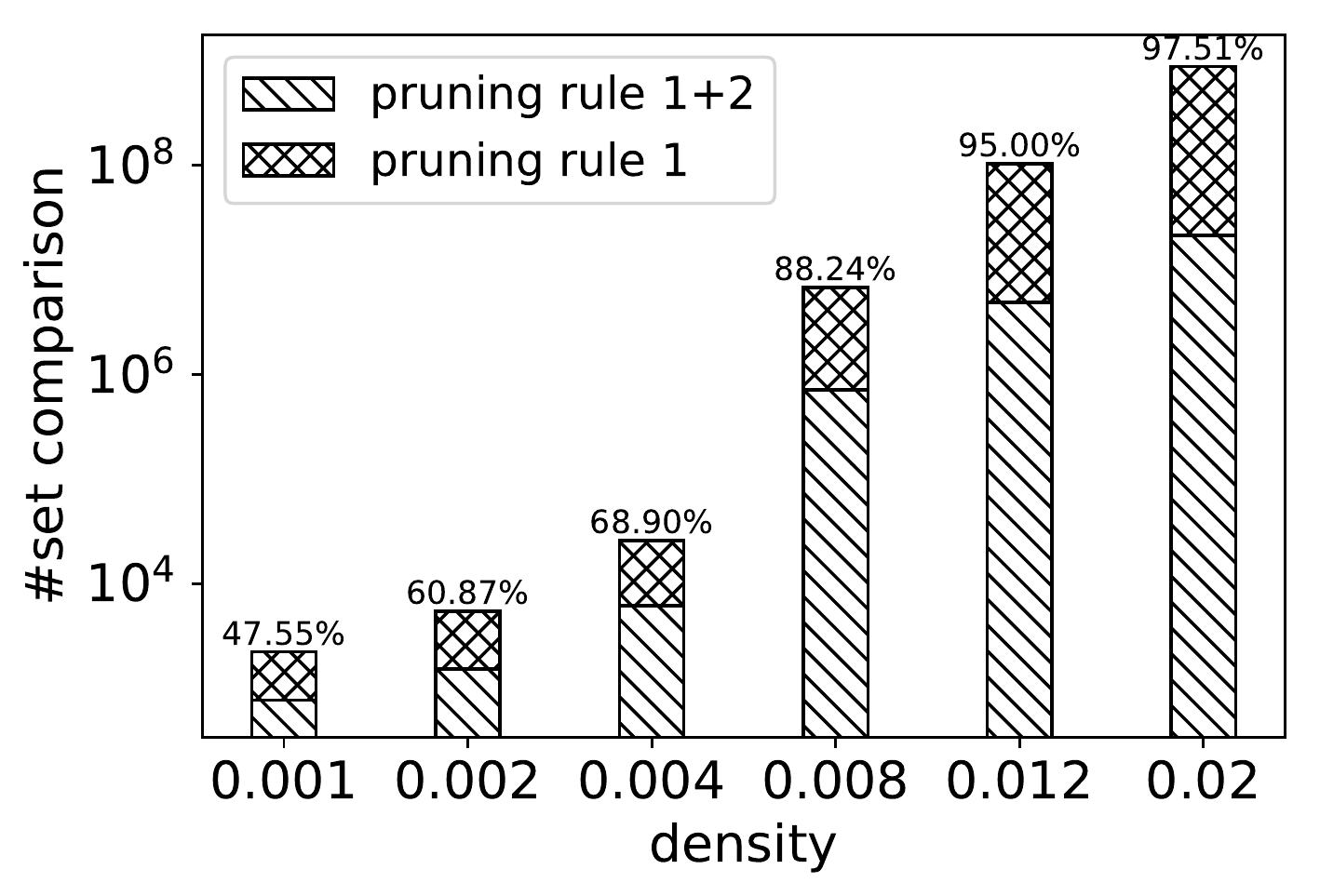}
    \end{minipage}}
    %\vspace{-1.5ex}
    \setlength{\belowcaptionskip}{-13pt}
    \caption{The effectiveness of pruning rule 2.} \label{fig:exp_pr}
\end{figure}

\subsubsection{Effectiveness of pruning rules}
As we have analyzed before, the bottleneck of time complexity for exact algorithms is set comparisons for local spatial clusters and the two pruning rules decrease time by reducing set comparisons at different levels. As we have shown, when $d$ or $density$ increase, the pruning rule 2 become more effective in reducing execution time. To further present the effectiveness of different pruning rules, we record the numbers of set comparisons when implementing only first pruning rule or both rules. Fig.~\ref{fig:exp_pr} shows the results on Gowalla and uniform synthetic datasets. Pruning rule 2 can help decrease the number of set comparisons by orders and when $d$ or $density$ increase, it is observed to reduce more set comparisons. When $d$ is set as 1 km on Gowalla dataset, pruning rule 2 can reduce more than $99\%$ set comparisons of  exact+rule1. 
\begin{figure}[htbp]
    \centering
    \subfigure[time vs. $d$ (Brightkite)]{\label{fig:exp_d_framework_1}
    \begin{minipage}[b]{0.15\textwidth}
        \centering 
        \includegraphics[width=\textwidth]{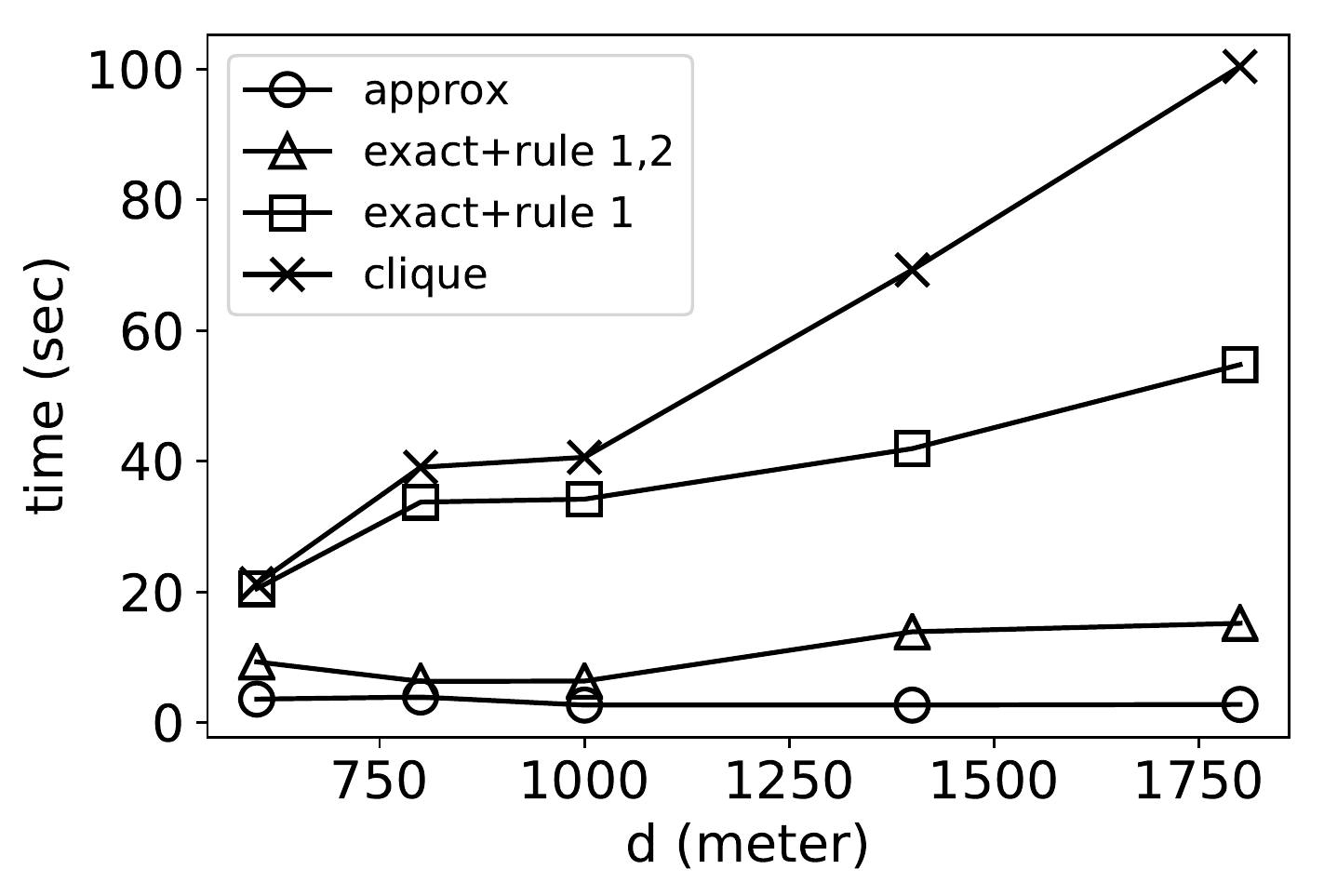}
    \end{minipage}}
    \subfigure[time vs. $d$ (Gowalla)]{\label{fig:exp_d_framework_2}
    \begin{minipage}[b]{0.15\textwidth}
        \centering 
        \includegraphics[width=\textwidth]{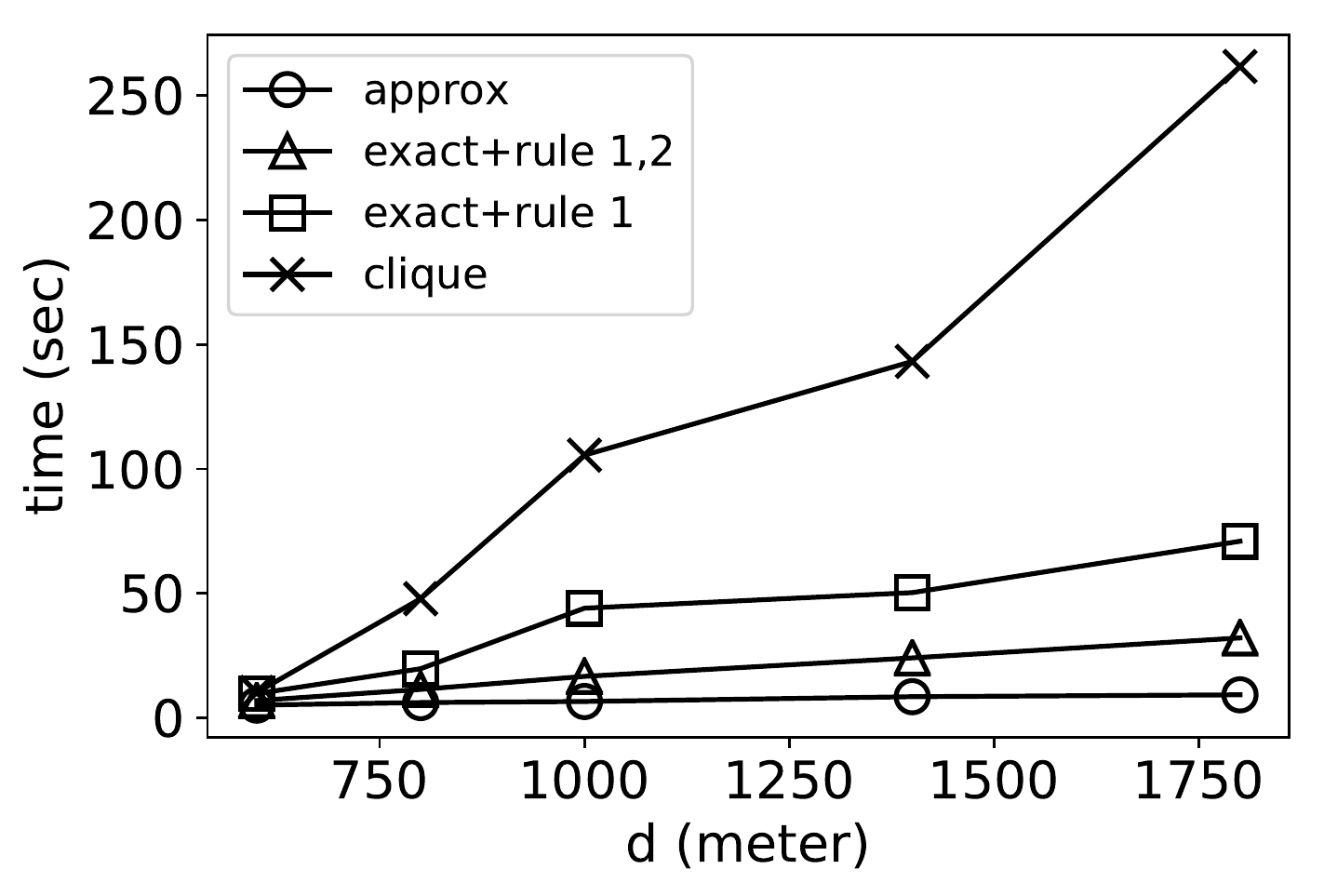}
    \end{minipage}}
    \subfigure[time vs. $d$ (Weibo)]{\label{fig:exp_d_framework_3}
    \begin{minipage}[b]{0.15\textwidth}
        \centering 
        \includegraphics[width=\textwidth]{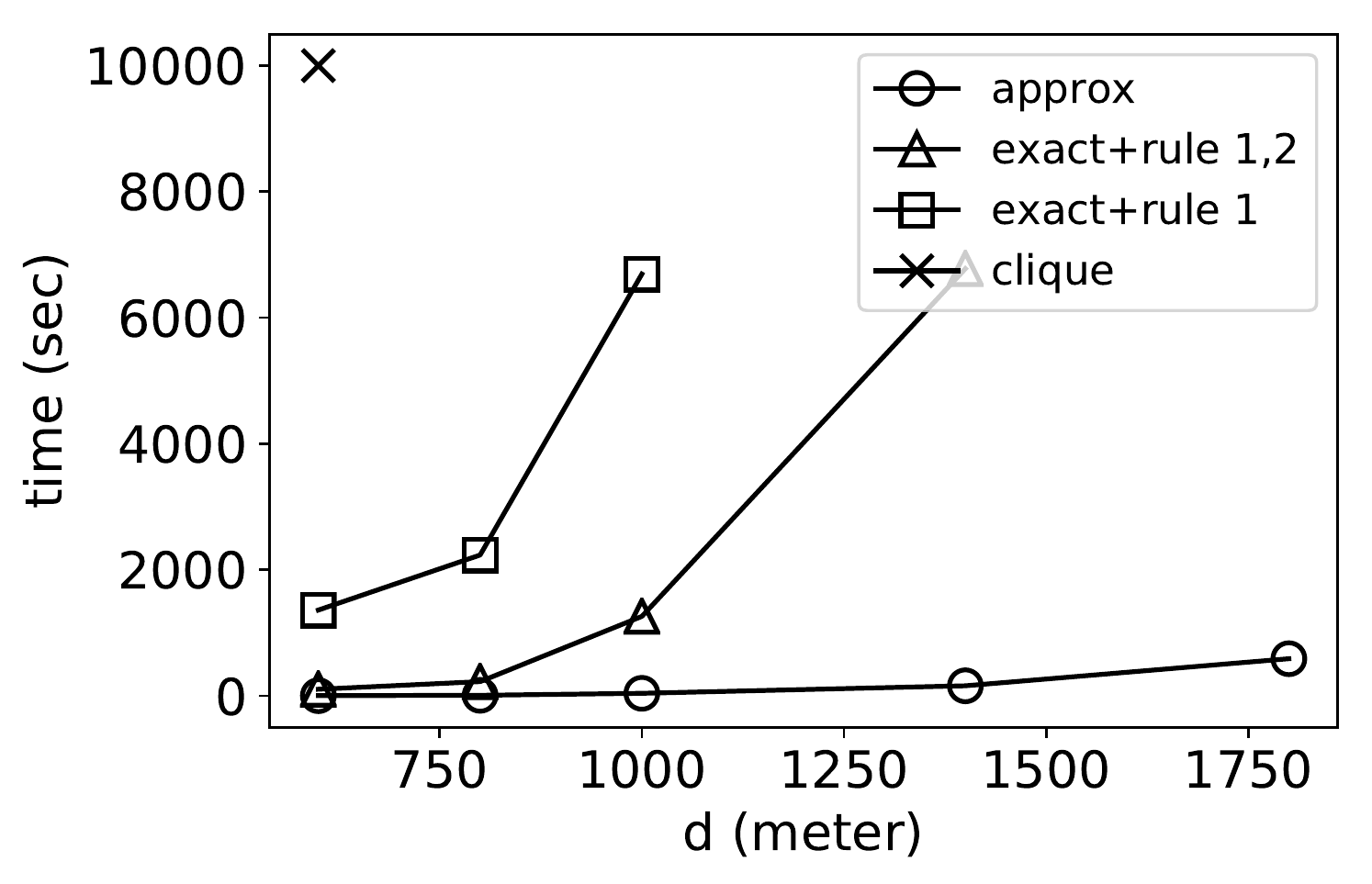}
    \end{minipage}}
    
    \subfigure[time vs. $k$ (Brightkite)]{\label{fig:exp_k_framework_1}
    \begin{minipage}[b]{0.15\textwidth}
        \centering 
        \includegraphics[width=\textwidth]{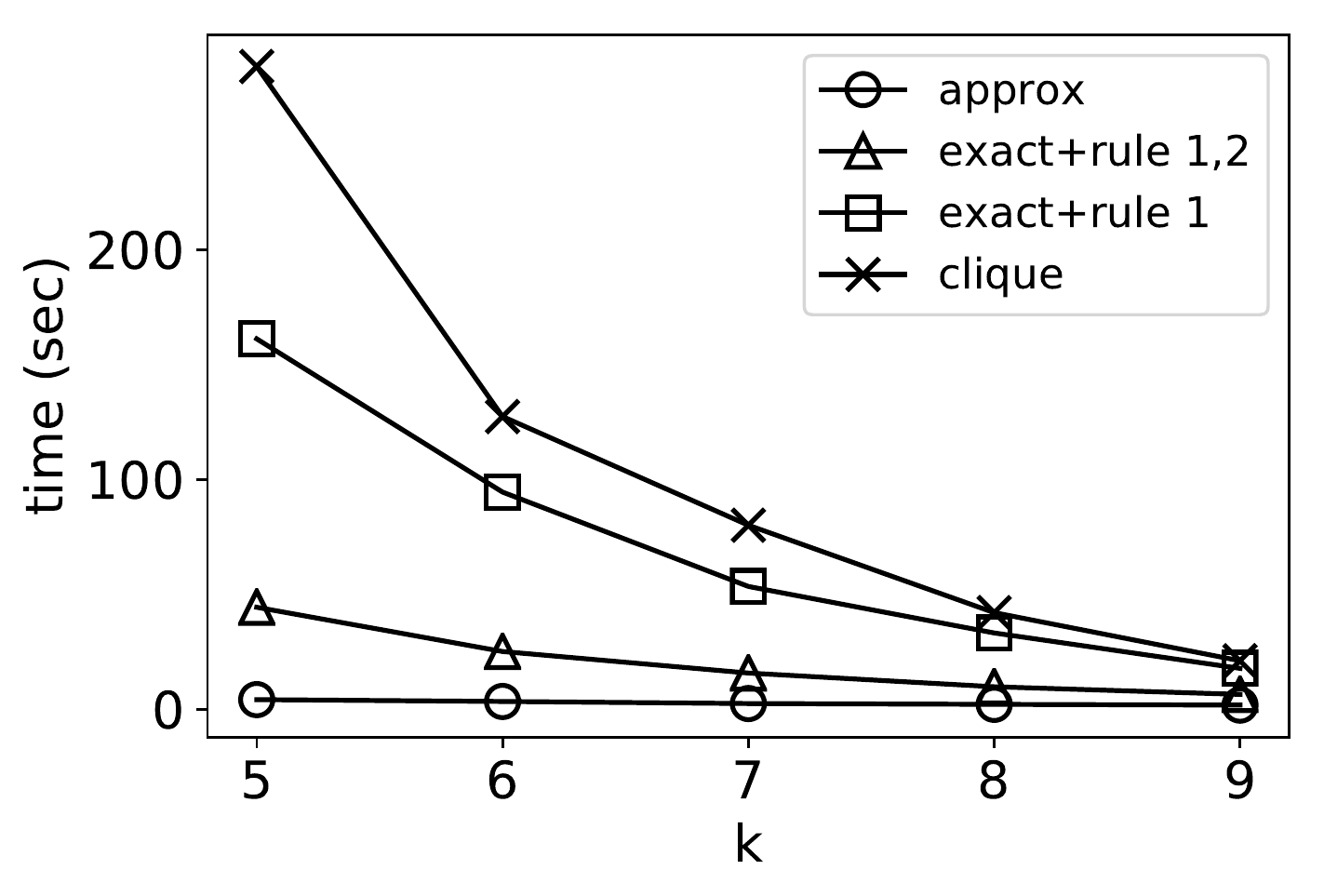}
    \end{minipage}}
    \subfigure[time vs. $k$ (Gowalla)]{\label{fig:exp_k_framework_2}
    \begin{minipage}[b]{0.15\textwidth}
        \centering 
        \includegraphics[width=\textwidth]{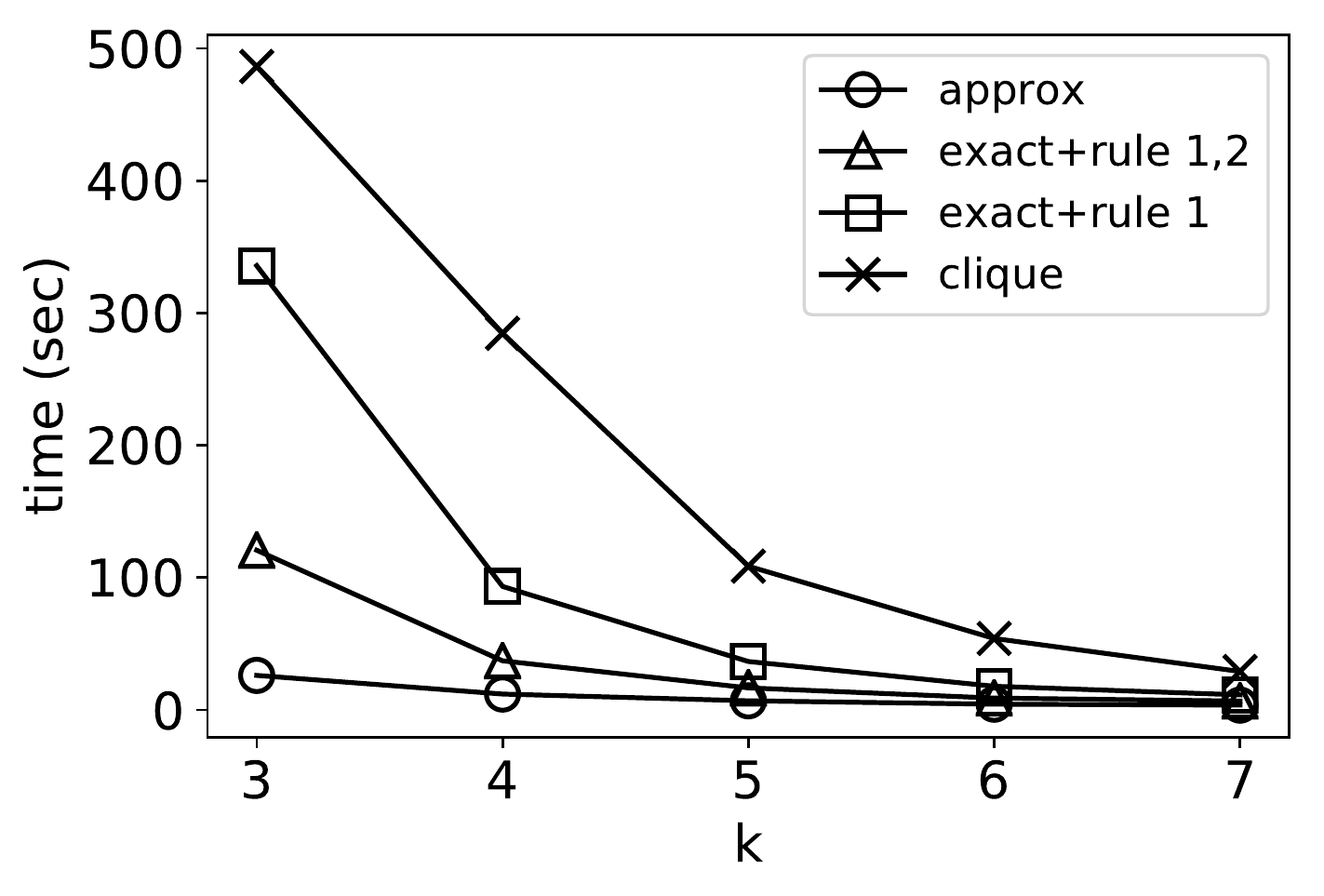}
    \end{minipage}}
    \subfigure[time vs. $k$ (Weibo)]{\label{fig:exp_k_framework_3}
    \begin{minipage}[b]{0.15\textwidth}
        \centering 
        \includegraphics[width=\textwidth]{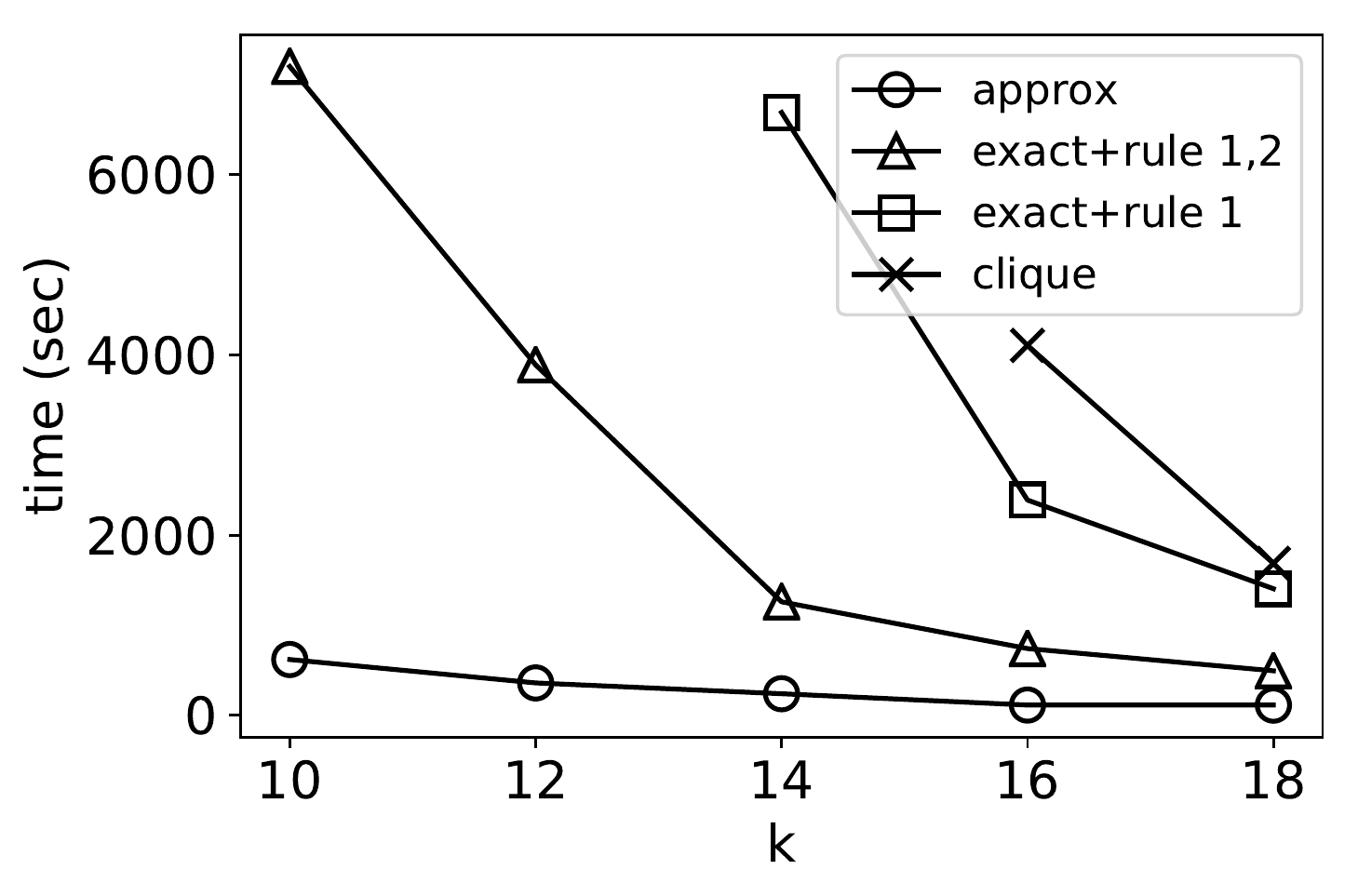}
    \end{minipage}}
    %\vspace{-1.5ex}
    %\setlength{\belowcaptionskip}{-13pt}
     \caption{MCC framework evaluation results of the influence of $k$ and $d$ on real geo-social networks.} 
     \label{fig:exp_kd_framework}
\end{figure}
\vspace{-1em}

\subsection{Framework Evaluation}
The previous subsection presents the results for spatial algorithm and this part will demonstrate the efficiency and effectiveness for the whole framework to detect all maximal co-located communities, which is shown in Algorithm~\ref{algo:framework}. The experiments are conducted on three real world geo-social networks and Table~\ref{tab:staticsSocial} shows the statistics of social network information. Note that before running the algorithms, we do some data cleaning works for original datasets, e.g., deleting all self-loop edges. 
\begin{table}[htbp]\caption{Statistics of social network information of real-world datasets.}\label{tab:staticsSocial}
\footnotesize
\begin{tabular}{|c|c|c|c|c|}
\hline
\textbf{Dataset} & \textbf{\#Vertices} & \textbf{\#Edges } & \textbf{Avg. Degree} &  \textbf{Max Degree} \\ \hline\hline
Brightkite & 58K & 214K & 7 & 1134 \\ \hline
Gowalla   &69K & 175K & 5 & 739\\ \hline
Weibo   & 1,019K & 8,245K & 16& 1100 \\ \hline
\end{tabular}
\end{table}
\vspace{-1.4em}
\begin{table}[htbp]\caption{Parameter setting for framework.}\label{tab:parameter_fram}
\footnotesize
\begin{tabular}{|c|c|c|}
\hline
\textbf{Dataset} & \textbf{$k$ values} & \textbf{$d$ values (meter)} \\ \hline\hline
Brightkite &  $[5, 6, \underline{7}, 8, 9]$ & $[600, 800, \underline{1000}, 1400, 1800]$      \\ \hline
Gowalla & $[3, 4, \underline{5}, 6, 7]$ & $[1000, 1250, \underline{1500}, 1750, 2000]$             \\ \hline
Weibo   & $[10, 12, \underline{14}, 16, 18]$ &    $[1000, 1250, \underline{1500}, 1750, 2000]$          \\ \hline
\end{tabular}
\end{table}
\vspace{-0.6em}

For social constraint in the framework, we implement both $k$-core and $k$-truss, however, due to the limit of page, we  only present the evaluation results of framework based on  $k$-core, and the results on $k$-truss have very similar performance. To make the framework more efficient, we adopt a simple pruning rule similar to the one used in \cite{chen2018maximum}. The pruning rule is based on the fact that a MCC must be a subset of a $k$-core (or $k$-truss), thus we first generate all $k$-cores from social network by applying core decomposition algorithm, and then apply our framework in each $k$-core to get all MCCs. Table~\ref{tab:parameter_fram} shows the settings of two parameters: $k$ (of $k$-core) and distance threshold $d$.
%\vspace{-1.5em}

%\vspace{-1.5em}
\subsubsection{Effect of $d$}
Fig.~\ref{fig:exp_d_framework_1}-\ref{fig:exp_d_framework_3} show the total execution time w.r.t. $d$. Since we apply spatial algorithm in each $k$-core instead of for all data points, the execution time for detecting MCCs is much less than that of detecting all spatial clusters presented in the last subsection. Clique is the slowest one on all datasets and increases dramatically w.r.t. $d$. The exact algorithms present efficiency on Brightkite and Gowalla datasets and the time does not increase much as $d$ increases. However, for Weibo dataset, time increases quickly with $d$. A possible reason is that data points in Weibo have much higher degree and there can be a $k$-core consisting of many  data points, and applying the spatial algorithm in that core can still be time consuming and the change of time w.r.t. $d$ is similar to the spatial algorithm experiment result as Fig.~\ref{fig:exp_d_weibo} shows.

\subsubsection{Effect of $k$}
Fig.~\ref{fig:exp_k_framework_1}-\ref{fig:exp_k_framework_3} present results on three real datasets by changing $k$. The execution time for all datasets decreases dramatically when $k$ turns larger. The reason is that, when $k$ increases, each $k$-core on social network would have smaller size, thus applying our framework on each $k$-core would save time. 

\subsubsection{Correctness of approximation algorithm} The above results have already demonstrated the efficiency and scalability of our approximation algorithm. To further validate its correctness, in each community detected by applying the approximation algorithm as spatial algorithm, we calculate the maximum pairwise distance as the community distance and Fig.\ref{fig:appro} present the average and maximum community distance of all communities. It shows that the community distance is always bounded by $\sqrt{2}\cot d$ and the average  distance is normally smaller than $d$ which means many communities have distance smaller than exact  threshold.
\begin{figure}[htbp]
    \centering
    \begin{minipage}[b]{0.23\textwidth}
        \centering 
        \includegraphics[width=\textwidth]{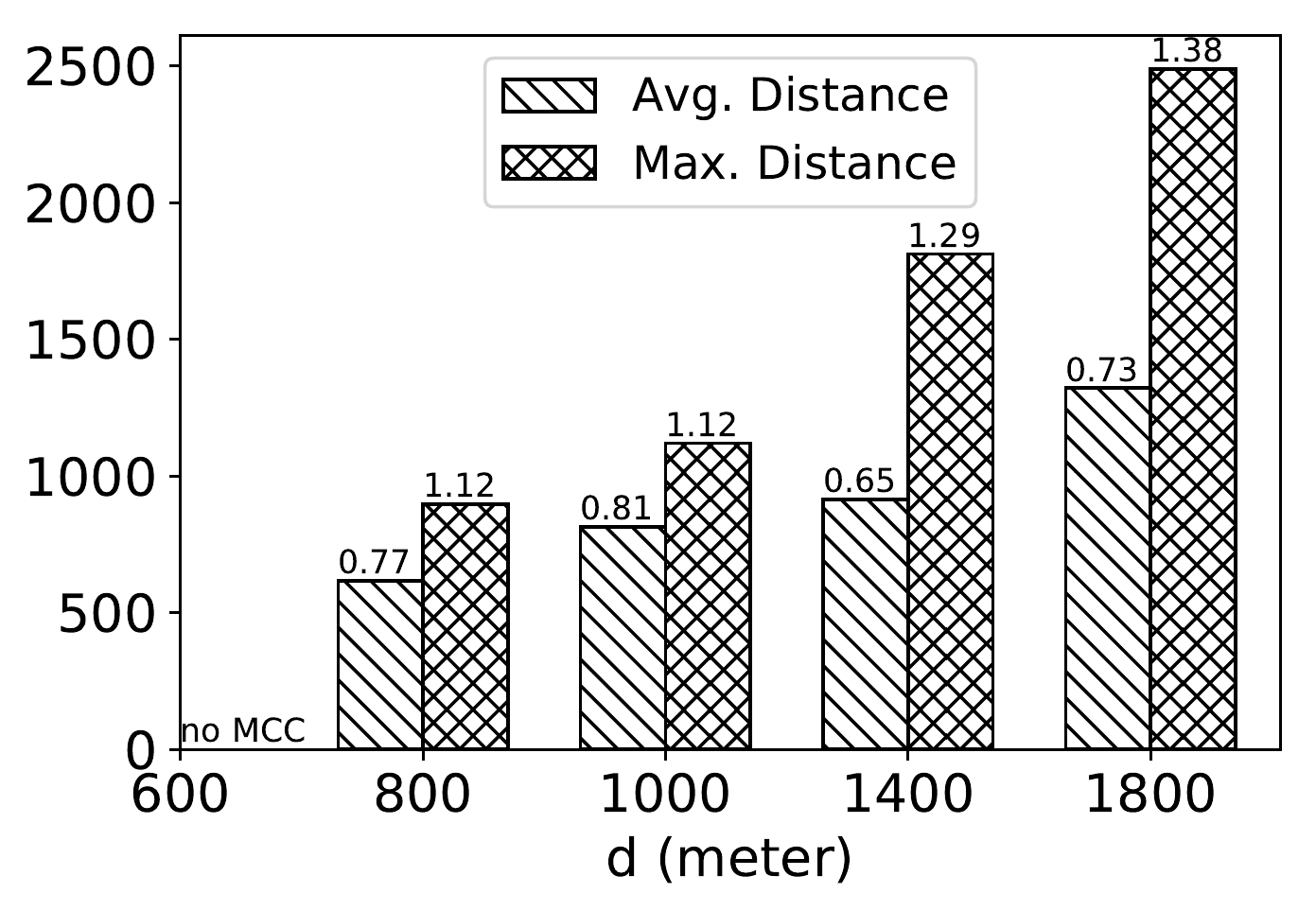}
    \end{minipage}
    \begin{minipage}[b]{0.23\textwidth}
        \centering 
        \includegraphics[width=\textwidth]{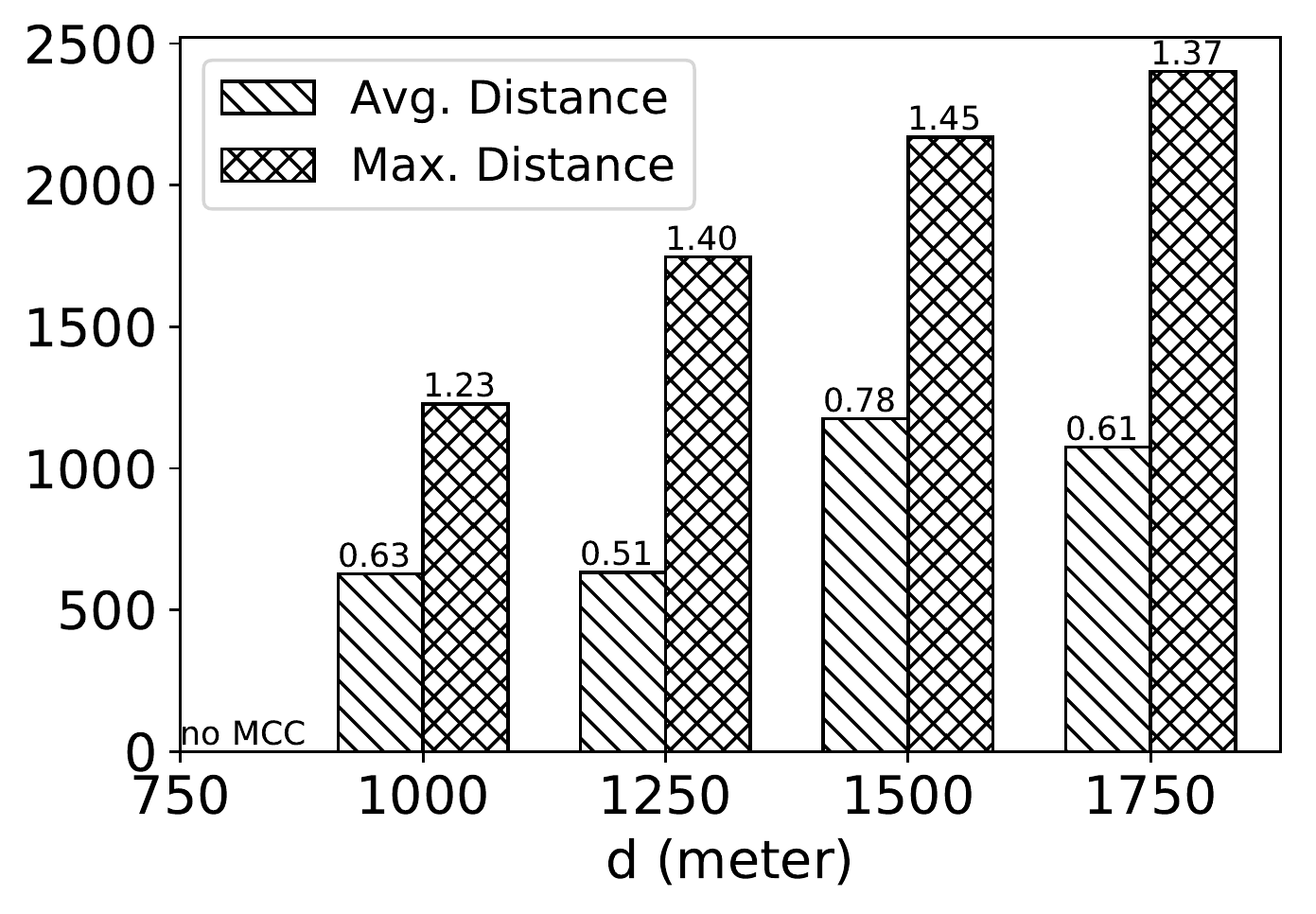}
    \end{minipage}
    %\vspace{-0.6em}
    \setlength{\belowcaptionskip}{-2pt}
    \caption{ Distance of MCCs on Brightkite and Gowalla.} \label{fig:appro}
\end{figure}
\subsection{Case Studies}
We implement the algorithms in \cite{fang2017effective, chen2018maximum}.  Both the two papers have different problem definitions with us. 
\cite{fang2017effective} provides a community search algorithm where the distance constraint is not defined in the same way as our work, 
and \cite{chen2018maximum} solves the problem to find only the maximum MCC and it applies all-pair distance constraint as spatial constraint. 
We conduct two case studies  on Gowalla and Brightkite  datasets by using our  approximation algorithm and  compare the  result with that of \cite{fang2017effective, chen2018maximum}  respectively to demonstrate the effectiveness of our problem and algorithm.

\subsubsection{Bounded Spatial Distance Guarantee} We  conduct experiment on Gowalla dataset and set  $k = 2$ and $d=2 km$. 
Fig.\ref{fig:case_study_1} (a) shows all MCCs detected by our algorithm in a small region. Each circle is the location of the MCC center and the color indicates the number of community members. There are 20 MCCs in this region. 
We also  present two communities  shown as the red  circles in (b) and (c)  respectively retrieved by using the community search algorithm in 
\cite{fang2017effective} with two  different query users.   In (b), the purple circle with diameter $d$ covers a MCC found by our algorithm. The method in \cite{fang2017effective} only returns a small subset of our MCC in order to make sure that the covering circle has the minimum radius. In (c),  \cite{fang2017effective} returns a community that has  a minimum covering circle with diameter much larger than $d = 2$km, and is not detected as a MCC by our algorithm. 
  As Fig.\ref{fig:case_study_1} shows,   \cite{fang2017effective} does not allow user to specify the distance threshold, and different  MCCs do not have consistent distance bound. For a query user who have many nearby friends,  \cite{fang2017effective} may return a small subset, however, for user who do not have nearby friends, it still returns a cluster with large distance among cluster members. 
 
\subsubsection{Diverse MCCs} On Brightkit dataset, by setting $d = 1$ km and  $k = 4$, we detect 32 MCCs. We conduct hierarchical cluster  analysis on 32  sets where Jaccard distance is used to measure the set distance. As Fig.\ref{fig:case_study_2} presents, there are five communities  that do not share any common user and   there are 9 communities when distance is set as 0.6. The results indicate that many MCCs  have diverse set members. However, the problem in  \cite{chen2018maximum} only find one maximum MCC and ignore all others despite the fact that other MCCs are equally meaningful and very different from members in the maximum MCC.

\begin{figure}[thbp]
    \centering
    \includegraphics[width=0.4\textwidth]{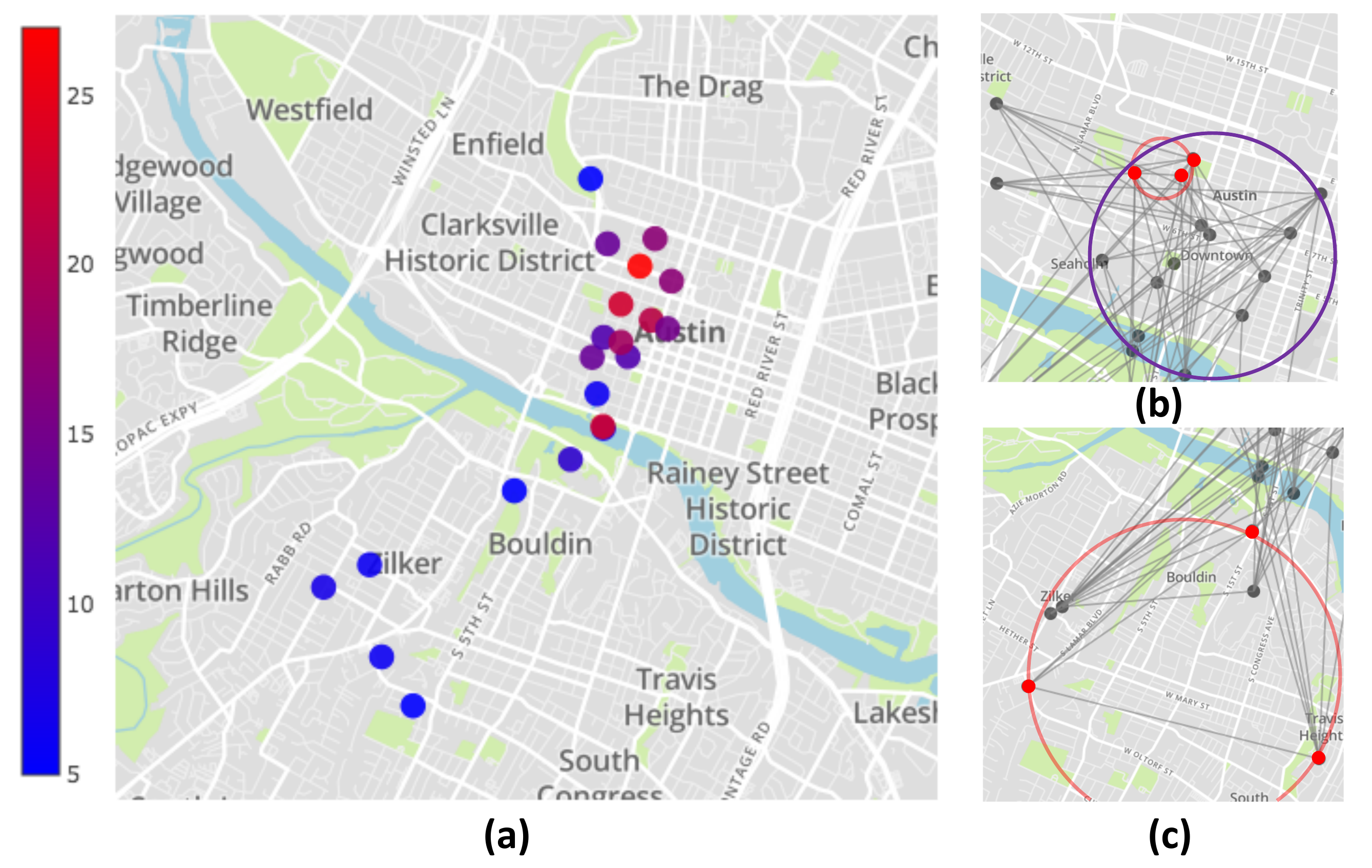}
    \setlength{\belowcaptionskip}{-17pt}
    \caption{2-MCCs detection  and results of \cite{fang2017effective} .} \label{fig:case_study_1}
\end{figure}
%%\vspace{-0.1ex}

\begin{figure}[htbp]
    \centering
    \includegraphics[width=0.47\textwidth]{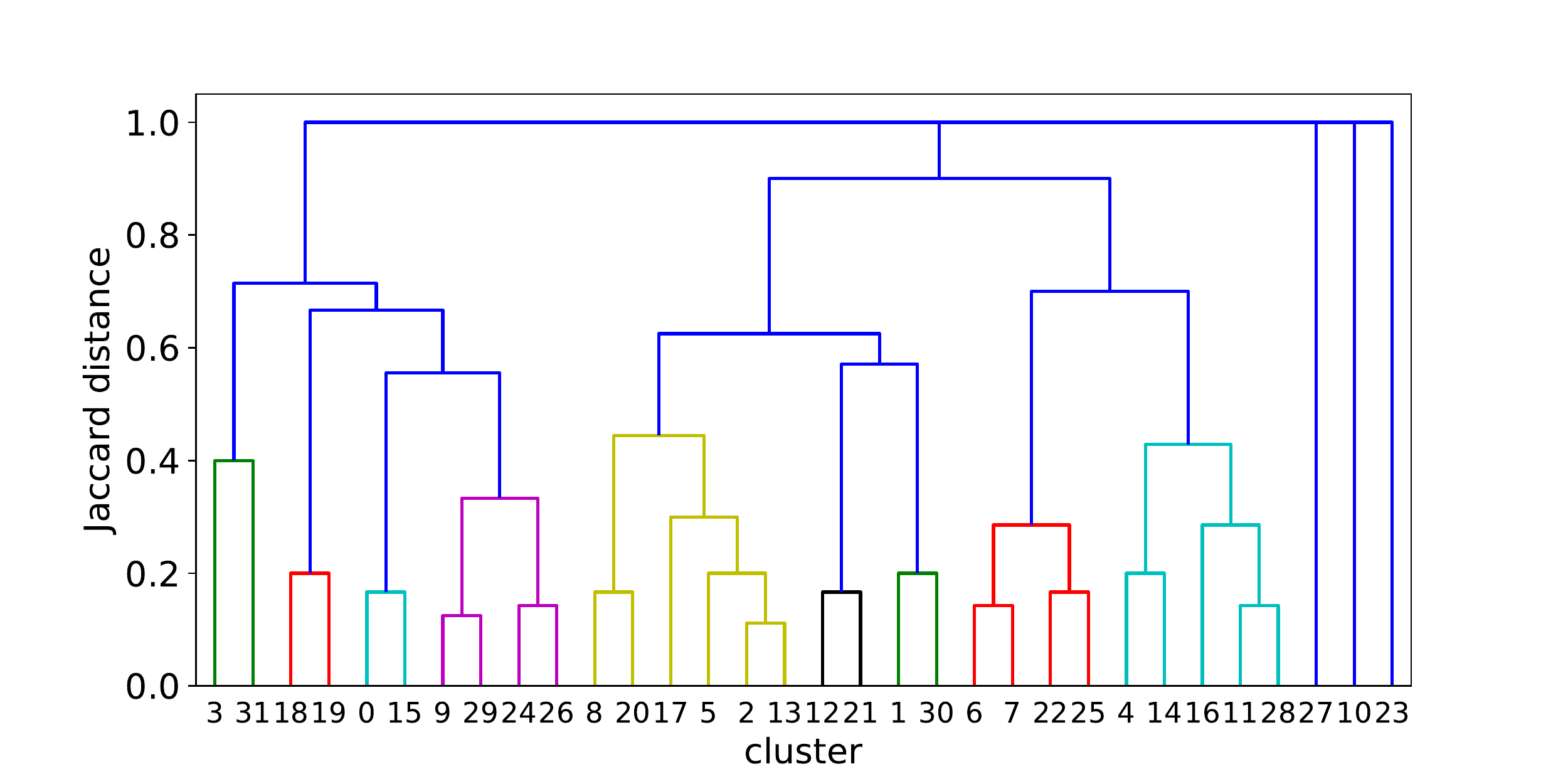}
    \setlength{\belowcaptionskip}{-2pt}
    \caption{Clustering Dendrogram for MCCs.} \label{fig:case_study_2}
\end{figure}

%% file: exp_results/secs/conclusion.tex
\section{Conclusion}\label{sec:conclusion}
In this paper, we investigate the $d$-MCCs detection problem on large scale geo-social networks. Unlike prior work that searches MCC for given query nodes or finds \emph{one} maximal MCC, we solve a community detection problem which detects \emph{all} communities satisfying both social and spatial cohesiveness constraints. To make our solution compatible with existing community detection techniques, we design a uniform framework so that existing techniques like $k$-core and $k$-truss decomposition can be easily plugged in. Besides generality and compatibility, our MCC detection framework improves efficiency thanks to our spatial constraint checking algorithms and several engineering level optimization. The effectiveness and efficiency of both the spatial algorithm and the whole MCC detection framework are demonstrated by using three real-world datasets and two synthetic datasets with various parameter settings. 